%% file: main.tex
\documentclass[11pt]{article}
\usepackage{amsmath,amsthm,amssymb}
\usepackage[numbers]{natbib}
\usepackage[margin=1in]{geometry}
\usepackage{cleveref}
\usepackage{thmtools}
\usepackage{subcaption}
\usepackage{adjustbox}
\usepackage{mathtools}

\usepackage{algorithm}
\crefname{algocf}{alg.}{algs.}
\Crefname{algocf}{Algorithm}{Algorithms}
\usepackage{algpseudocode}
\usepackage[mode=buildnew]{standalone}

\usepackage{tikz}
\usetikzlibrary{graphs,graphs.standard}
\usetikzlibrary{positioning,shapes.geometric}
\usetikzlibrary{arrows.meta}
\usepackage{enumitem}

\usepackage{xcolor}
\usepackage{todonotes}

\newtheorem{thm}{Theorem}
\newtheorem{lem}[thm]{Lemma}
\newtheorem{cor}[thm]{Corollary}

\newtheorem{observation}{Observation}

\newtheorem*{defn}{Definition}

\usepackage{tabularx,lipsum,environ}
\makeatletter
\newcommand{\problemtitle}[1]{\gdef\@problemtitle{#1}}
\newcommand{\probleminput}[1]{\gdef\@probleminput{#1}}
\newcommand{\problemquestion}[1]{\gdef\@problemquestion{#1}}
\NewEnviron{problem}{
  \problemtitle{}\probleminput{}\problemquestion{}
  \BODY
  \par\addvspace{.5\baselineskip}
  \noindent
  \begin{tabularx}{\textwidth}{@{\hspace{\parindent}} l X c}
    \multicolumn{2}{@{\hspace{\parindent}}l}{\@problemtitle} \\
    \textbf{Input:} & \@probleminput \\
    \textbf{Question:} & \@problemquestion
  \end{tabularx}
  \par\addvspace{.5\baselineskip}
}
\makeatother

\title{Eliminating Majority Illusion is Easy}

\author{Jack Dippel\thanks{McGill University: {\tt jack.dippel@mail.mcgill.ca}}
\and
Max Dupré la Tour\thanks{McGill University:
{\tt max.duprelatour@mail.mcgill.ca}}
\and
April Niu\thanks{McGill University: {\tt yuexing.niu@mail.mcgill.ca}}
 \and
Sanjukta Roy\thanks{University of Leeds: {\tt s.roy@leeds.ac.uk}}
\and
Adrian Vetta\thanks{McGill University: {\tt adrian.vetta@mcgill.ca}}
}

\date{}
\begin{document}

\bibliographystyle{plainnat}

\renewcommand{\thefootnote}{\arabic{footnote}}

\maketitle

\begin{abstract}
Majority Illusion is a phenomenon in social networks wherein the decision by the majority of the network is not the same as one's personal social circle's majority, leading to an incorrect perception of the majority in a large network. In this paper, we present polynomial-time algorithms which can eliminate majority illusion in a network by altering as few connections as possible. Additionally, we prove that the more general problem of ensuring all neighbourhoods in the network are at least a $p$-fraction of the majority is NP-hard for most values of $p$.
\end{abstract}

\input{introduction}

\input{preliminaries}

\input{add}

\input{sub}

\input{both}

\input{hardness}

\input{conclusion}

\bibliography{main}
\end{document}

%% file: introduction.tex
\section{Introduction}\label{sec:intro}

Perceptions are liable to become distorted on a social network.
An exemplar of this is the {\em friendship paradox}~\cite{Fel91}: on average, a person has fewer
friends than their friends do.\footnote{This is a simple consequence of the fact that a higher degree node appears in more neighbourhoods than a lower degree node.}
Similarly, on average, a researcher is less productive than their co-authors~\cite{EJ14, BLA16}.
This network characteristic, whereby local observations can lead to erroneous conclusions regarding a global state,  
was dubbed {\em majority illusion} by Lerman et al.~\cite{Lerman2016}. This nomenclature relates to situations
where a minority viewpoint is perceived to be the majority viewpoint~\cite{Stewart2019}.



One of the simplest manifestation of majority illusion concerns 
election. Let's model a two-party election by a graph $G=(V,E)$
where each node represents a voter and each edge connects two voters that are neighbours (or friends).
Furthermore, imagine party preference is given by a colouring function $f:V\rightarrow \{B, R\}$. If there are more blue nodes than red nodes then we say blue is the {\em majority}.
A node $v\in V$ is then under {\em (majority) illusion} if it has strictly more red neighbours than blue neighbours.
It is easy to construct instances where some nodes are under majority illusion; see Figure~\ref{fig:examples}(a).
Much more surprising is the existence of instances where {\em every} node is under illusion; see Figure~\ref{fig:examples}(b). 

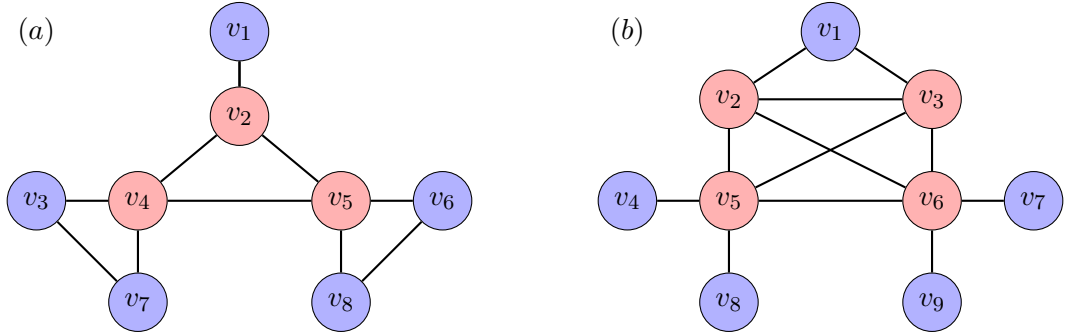
\begin{figure}[ht!]
    \centering
\tikzstyle{node}=[ draw=black,
 shape=circle, minimum size=10pt]
 \tikzstyle{edge}=[thick, -] 
\begin{tikzpicture}[scale=0.45]
\node at (-6,5) {$(a)$};
\node [style=node,fill=blue!30] (V8) at (3,-3) {$v_8$};
\node [style=node,fill=red!30] (V4) at (-3, 0) {$v_4$};
 \node [style=node,fill=red!30] (V2) at (0, 2.5) {$v_2$};
 \node [style=node,fill=blue!30] (V1) at (0, 5) {$v_1$};
\node [style=node,fill=red!30] (V5) at (3, 0) {$v_5$};
\node [style=node,fill=blue!30] (V6) at (6, 0) {$v_6$};
\node [style=node,fill=blue!30] (V3) at (-6, 0) {$v_3$};
\node [style=node,fill=blue!30] (V7) at (-3, -3) {$v_7$};
    \draw [style=edge] (V4) -- (V3);
    \draw [style=edge] (V4) -- (V7);
    \draw [style=edge] (V5) -- (V8);
    \draw [style=edge] (V5) -- (V6);
    \draw [style=edge] (V4) -- (V2);
    \draw [style=edge] (V2) -- (V5);
    \draw [style=edge] (V2) -- (V1);
    \draw [style=edge] (V8) -- (V6);
    \draw [style=edge] (V7) -- (V3);
    \draw [style=edge] (V1) -- (V2);
    \draw [style=edge] (V4) -- (V5);
\end{tikzpicture}
\qquad\qquad
\begin{tikzpicture}[scale=0.45]
\node at (-6,5) {$(b)$};
\node [style=node,fill=blue!30] (V9) at (3,-3) {$v_9$};
\node [style=node,fill=red!30] (V5) at (-3, 0) {$v_5$};
 \node [style=node,fill=red!30] (V3) at (3, 3) {$v_3$};
 \node [style=node,fill=blue!30] (V1) at (0, 5) {$v_1$};
 \node [style=node,fill=red!30] (V2) at (-3, 3) {$v_2$};
\node [style=node,fill=red!30] (V6) at (3, 0) {$v_6$};
\node [style=node,fill=blue!30] (V7) at (6, 0) {$v_7$};
\node [style=node,fill=blue!30] (V4) at (-6, 0) {$v_4$};
\node [style=node,fill=blue!30] (V8) at (-3, -3) {$v_8$};
    \draw [style=edge] (V5) -- (V4);
    \draw [style=edge] (V5) -- (V8);
    \draw [style=edge] (V6) -- (V9);
    \draw [style=edge] (V6) -- (V7);
    \draw [style=edge] (V5) -- (V3);
    \draw [style=edge] (V3) -- (V6);
    \draw [style=edge] (V5) -- (V2);
    \draw [style=edge] (V2) -- (V6);
    \draw [style=edge] (V3) -- (V1);
    \draw [style=edge] (V2) -- (V1);
      \draw [style=edge] (V2) -- (V3);
        \draw [style=edge] (V5) -- (V6);
\end{tikzpicture}
    \caption{In network (a) there are two nodes, $v_1$ and $v_2$, under majority illusion; in network (b) every node suffers illusion!}
    \label{fig:examples}
\end{figure}

Network topology can have a big impact on misinformation. 
Indeed, electoral networks are very susceptible to majority illusion 
because their underlying network topology has a social determinant.
For example, {\em homophily}, the tendency for humans to connect with similar people, contributes to the creation of information bubbles and polarization~\cite{McPherson2001}.

Another topical example of majority illusion concerns vaccinations~\cite{Johnson2020}.
Here node colours may indicate vaccinated or unvaccinated.
If a personal decision to vaccinate oneself is influenced by 
the actions of people in your local social network
then misinformation here may have significant consequences.
Moreover, a smaller number of influential high degree nodes
may suffice to cause major distortions.
An illusion can occur even when there are more than two alternatives and each node selects one alternative. Then, a node under illusion has strictly less blue neighbors (i.e., the neighbors select the winning alternative) compared to red neighbors (denotes the set of all non-winning alternatives).
Indeed, this observation indicates the relevance of majority illusion to advertising, campaigning, and influence maximization~\cite{becker2023improving,zhou2021maximizing}.
Moreover, the effects of illusion have been studied in many disparate areas, including voting theory~\cite{bara2021predicting,castiglioni2020election,doucette2019inferring,wilder2018controlling}, participatory budgeting and the allocation of public goods~\cite{kempe2020inducing,yu2021altruism}. 



These applications highlighting the importance of creating social networks that give everyone fair and unbiased information~\cite{Johnson2020}.
That is the focus of this paper.
Specifically, to determine whether or not majority illusion can be eliminated in a network? If so, can we do this optimally in polynomial time?

\subsection{The Illusion Elimination Problem}
 To answer these questions we first study a class of {\em majority illusion elimination problems}. 
 We are given a social network $G=(V,E)$, a non-negative integer $k$, and
 a voting function $f:V\rightarrow \{B,R\}$, where blue is the strict \emph{majority}. Can we eliminate (majority) illusion, for every node, by altering (adding and/or deleting) at most $k$ edges?
That is, after the alteration, no node can have a strict majority of
red nodes in its neighbourhood.
 
 Observe that, since we may remove edges of $G$ or add non-edges of $G$,
there are three variants of the majority illusion elimination problems
where we can only add edges, only remove edges, or both add and remove edges.
Formally, this gives the following three decision problems, respectively:

\begin{problem}
  \problemtitle{\sc Majority Illusion Addition Elimination (MIAE)}
  \probleminput{A graph $G = (V,E)$, a function $f : V \mapsto \{B,R\}$, and a non-negative integer $k$.}
  \problemquestion{Is there an edge set $E'$ on $V$ with $E \subset E', |E'\setminus E| \leq k$ such that no node is under illusion in $G' = (V,E')$ under $f$?}
\end{problem}

\begin{problem}
  \problemtitle{\sc Majority Illusion Removal Elimination (MIRE)}
  \probleminput{A graph $G = (V,E)$, a function $f : V \mapsto \{B,R\}$, and a non-negative integer $k$.}
  \problemquestion{Is there an edge set $E'$ on $V$ with $E' \subseteq E, |E\setminus E'| \leq k$ such that no node is under illusion in $G' = (V,E')$ under $f$?}
\end{problem}

\begin{problem}
  \problemtitle{\sc Majority Illusion Elimination (MIE)}
  \probleminput{A graph $G = (V,E)$, a function $f : V \mapsto \{B,R\}$, and a non-negative integer $k$.}
  \problemquestion{Is there an edge set $E'$ on $V$ with $|E\setminus E'|+|E'\setminus E| \leq k$ such that no node is under illusion in $G' = (V,E')$ under $f$?}
\end{problem}

It may be helpful to illustrate these three problems using the two examples we encountered in Figure~\ref{fig:examples}.
Figure~\ref{fig:semi-illusion} considers the example from Figure~\ref{fig:examples}(a) where only
two nodes, $v_1$ and $v_2$, are under majority illusion. It presents
optimal solutions that require the minimum number of edge alterations
for the three problems (that is, solutions to 
optimization versions of the above decision problems).
Similarly, Figure~\ref{fig:complete-illusion} considers the example Figure~\ref{fig:examples}(b) where every node is under illusion and gives the three corresponding optimal 
solutions to eliminate illusion.

\begin{figure}[ht!]
    \centering
\tikzstyle{node}=[ draw=black,
 shape=circle, minimum size=10pt]
 \tikzstyle{edge}=[thick, -]
\begin{tikzpicture}[scale=0.45]
\node at (-6,5) {$(a)$};
\node [style=node,fill=blue!30] (V8) at (3,-3) {$v_8$};
\node [style=node,fill=red!30] (V4) at (-3, 0) {$v_4$};
 \node [style=node,fill=red!30] (V2) at (0, 2.5) {$v_2$};
 \node [style=node,fill=blue!30] (V1) at (0, 5) {$v_1$};
\node [style=node,fill=red!30] (V5) at (3, 0) {$v_5$};
\node [style=node,fill=blue!30] (V6) at (6, 0) {$v_6$};
\node [style=node,fill=blue!30] (V3) at (-6, 0) {$v_3$};
\node [style=node,fill=blue!30] (V7) at (-3, -3) {$v_7$};
    \draw [style=edge] (V4) -- (V3);
    \draw [style=edge] (V4) -- (V7);
    \draw [style=edge] (V5) -- (V8);
    \draw [style=edge] (V5) -- (V6);
    \draw [style=edge] (V4) -- (V2);
    \draw [style=edge] (V2) -- (V5);
    \draw [style=edge] (V2) -- (V1);
    \draw [style=edge] (V8) -- (V6);
    \draw [style=edge] (V7) -- (V3);
    \draw [style=edge] (V1) -- (V2);
    \draw [style=edge] (V4) -- (V5);
\end{tikzpicture}
\hspace{10mm}
\begin{tikzpicture}[scale=0.45]
\node at (-6,5) {$(b)$};
\node [style=node,fill=blue!30] (V8) at (3,-3) {$v_8$};
\node [style=node,fill=red!30] (V4) at (-3, 0) {$v_4$};
 \node [style=node,fill=red!30] (V2) at (0, 2.5) {$v_2$};
 \node [style=node,fill=blue!30] (V1) at (0, 5) {$v_1$};
\node [style=node,fill=red!30] (V5) at (3, 0) {$v_5$};
\node [style=node,fill=blue!30] (V6) at (6, 0) {$v_6$};
\node [style=node,fill=blue!30] (V3) at (-6, 0) {$v_3$};
\node [style=node,fill=blue!30] (V7) at (-3, -3) {$v_7$};
    \draw [style=edge] (V4) -- (V3);
    \draw [style=edge] (V4) -- (V7);
    \draw [style=edge] (V5) -- (V8);
    \draw [style=edge] (V5) -- (V6);
    \draw [style=edge] (V4) -- (V2);
    \draw [style=edge] (V2) -- (V5);
    \draw [style=edge] (V2) -- (V1);
    \draw [style=edge] (V8) -- (V6);
    \draw [style=edge] (V7) -- (V3);
    \draw [style=edge] (V1) -- (V2);
    \draw [style=edge] (V4) -- (V5);
    \draw [style=edge] (V3) -- (V2);
    \draw [style=edge] (V3) -- (V1);
\end{tikzpicture}\\
\vspace{10mm}
\begin{tikzpicture}[scale=0.45]
\node at (-6,5) {$(c)$};
\node [style=node,fill=blue!30] (V8) at (3,-3) {$v_8$};
\node [style=node,fill=red!30] (V4) at (-3, 0) {$v_4$};
 \node [style=node,fill=red!30] (V2) at (0, 2.5) {$v_2$};
 \node [style=node,fill=blue!30] (V1) at (0, 5) {$v_1$};
\node [style=node,fill=red!30] (V5) at (3, 0) {$v_5$};
\node [style=node,fill=blue!30] (V6) at (6, 0) {$v_6$};
\node [style=node,fill=blue!30] (V3) at (-6, 0) {$v_3$};
\node [style=node,fill=blue!30] (V7) at (-3, -3) {$v_7$};
    \draw [style=edge] (V4) -- (V3);
    \draw [style=edge] (V4) -- (V7);
    \draw [style=edge] (V5) -- (V8);
    \draw [style=edge] (V5) -- (V6);
    \draw [style=edge] (V8) -- (V6);
    \draw [style=edge] (V7) -- (V3);
    \draw [style=edge] (V4) -- (V5);
\end{tikzpicture}
\hspace{10mm}
\begin{tikzpicture}[scale=0.45]
\node at (-6,5) {$(d)$};
\node [style=node,fill=blue!30] (V8) at (3,-3) {$v_8$};
\node [style=node,fill=red!30] (V4) at (-3, 0) {$v_4$};
 \node [style=node,fill=red!30] (V2) at (0, 2.5) {$v_2$};
 \node [style=node,fill=blue!30] (V1) at (0, 5) {$v_1$};
\node [style=node,fill=red!30] (V5) at (3, 0) {$v_5$};
\node [style=node,fill=blue!30] (V6) at (6, 0) {$v_6$};
\node [style=node,fill=blue!30] (V3) at (-6, 0) {$v_3$};
\node [style=node,fill=blue!30] (V7) at (-3, -3) {$v_7$};
    \draw [style=edge] (V4) -- (V3);
    \draw [style=edge] (V4) -- (V7);
    \draw [style=edge] (V5) -- (V8);
    \draw [style=edge] (V5) -- (V6);
    \draw [style=edge] (V2) -- (V5);
    \draw [style=edge] (V2) -- (V1);
    \draw [style=edge] (V8) -- (V6);
    \draw [style=edge] (V7) -- (V3);
    \draw [style=edge] (V1) -- (V2);
    \draw [style=edge] (V4) -- (V5);
    \draw [style=edge] (V1) -- (V3);
\end{tikzpicture}
    \caption{In (a) only $v_1$ and $v_2$ are under majority illusion. An optimal solution to the MIAE problem, adding just $2$ edges, is shown in (b); an optimal solution to the MIRE problem, removing $5$ edges, is shown in (c); an optimal solution to the MIE problem, adding $1$ edge and removing $1$ edge, is shown in (d).}
    \label{fig:semi-illusion}
\end{figure}
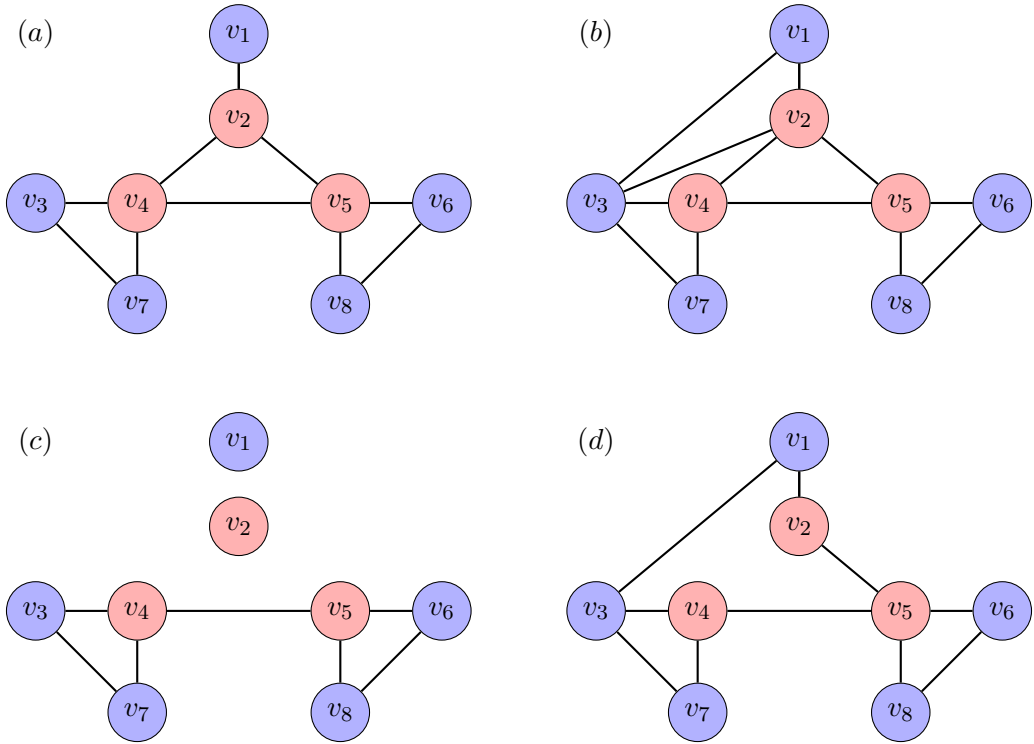

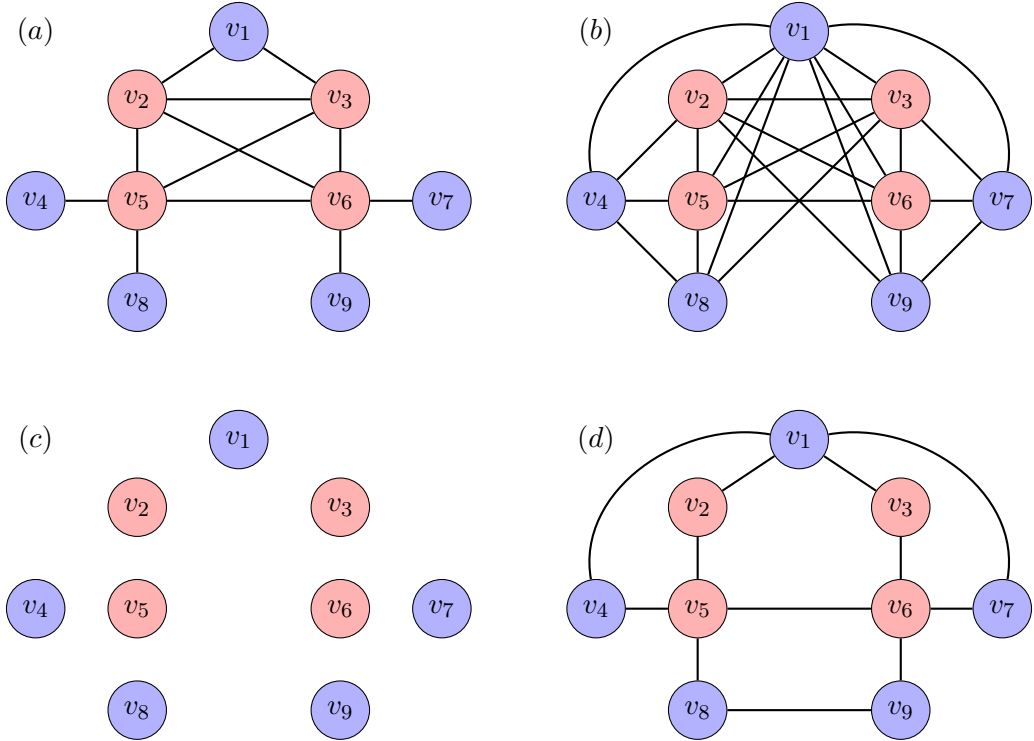
\begin{figure}[ht!]
    \centering
\tikzstyle{node}=[ draw=black,
 shape=circle, minimum size=10pt]
 \tikzstyle{edge}=[thick, -]
\begin{tikzpicture}[scale=0.45]
\node at (-6,5) {$(a)$};
\node [style=node,fill=blue!30] (V9) at (3,-3) {$v_9$};
\node [style=node,fill=red!30] (V5) at (-3, 0) {$v_5$};
 \node [style=node,fill=red!30] (V3) at (3, 3) {$v_3$};
 \node [style=node,fill=blue!30] (V1) at (0, 5) {$v_1$};
 \node [style=node,fill=red!30] (V2) at (-3, 3) {$v_2$};
\node [style=node,fill=red!30] (V6) at (3, 0) {$v_6$};
\node [style=node,fill=blue!30] (V7) at (6, 0) {$v_7$};
\node [style=node,fill=blue!30] (V4) at (-6, 0) {$v_4$};
\node [style=node,fill=blue!30] (V8) at (-3, -3) {$v_8$};
    \draw [style=edge] (V5) -- (V4);
    \draw [style=edge] (V5) -- (V8);
    \draw [style=edge] (V6) -- (V9);
    \draw [style=edge] (V6) -- (V7);
    \draw [style=edge] (V5) -- (V3);
    \draw [style=edge] (V3) -- (V6);
    \draw [style=edge] (V5) -- (V2);
    \draw [style=edge] (V2) -- (V6);
    \draw [style=edge] (V3) -- (V1);
    \draw [style=edge] (V2) -- (V1);
      \draw [style=edge] (V2) -- (V3);
        \draw [style=edge] (V5) -- (V6);
\end{tikzpicture}
\hspace{10mm}
\begin{tikzpicture}[scale=0.45]
\node at (-6,5) {$(b)$};
\node [style=node,fill=blue!30] (V9) at (3,-3) {$v_9$};
\node [style=node,fill=red!30] (V5) at (-3, 0) {$v_5$};
 \node [style=node,fill=red!30] (V3) at (3, 3) {$v_3$};
 \node [style=node,fill=blue!30] (V1) at (0, 5) {$v_1$};
 \node [style=node,fill=red!30] (V2) at (-3, 3) {$v_2$};
\node [style=node,fill=red!30] (V6) at (3, 0) {$v_6$};
\node [style=node,fill=blue!30] (V7) at (6, 0) {$v_7$};
\node [style=node,fill=blue!30] (V4) at (-6, 0) {$v_4$};
\node [style=node,fill=blue!30] (V8) at (-3, -3) {$v_8$};
    \draw [style=edge] (V5) -- (V4);
    \draw [style=edge] (V5) -- (V8);
    \draw [style=edge] (V6) -- (V9);
    \draw [style=edge] (V6) -- (V7);
    \draw [style=edge] (V5) -- (V3);
    \draw [style=edge] (V3) -- (V6);
    \draw [style=edge] (V5) -- (V2);
    \draw [style=edge] (V2) -- (V6);
    \draw [style=edge] (V3) -- (V1);
    \draw [style=edge] (V2) -- (V1);
    \draw [style=edge] (V2) -- (V3);
    \draw [style=edge] (V5) -- (V6);
    \draw [style=edge] (V5) -- (V1);
    \draw [style=edge] (V6) -- (V1);
    \draw [style=edge] (V8) -- (V1);
    \draw [style=edge] (V9) -- (V1);
    \draw [style=edge] (V4) to[out=100,in=170] (V1);
    \draw [style=edge] (V1) to[out=10,in=80] (V7);
    \draw [style=edge] (V4) -- (V2);
    \draw [style=edge] (V7) -- (V3);
    \draw [style=edge] (V9) -- (V2);
    \draw [style=edge] (V8) -- (V3);
    \draw [style=edge] (V9) -- (V7);
    \draw [style=edge] (V8) -- (V4);
\end{tikzpicture}\\
\vspace{10mm}
\begin{tikzpicture}[scale=0.45]
\node at (-6,5) {$(c)$};
\node [style=node,fill=blue!30] (V9) at (3,-3) {$v_9$};
\node [style=node,fill=red!30] (V5) at (-3, 0) {$v_5$};
 \node [style=node,fill=red!30] (V3) at (3, 3) {$v_3$};
 \node [style=node,fill=blue!30] (V1) at (0, 5) {$v_1$};
 \node [style=node,fill=red!30] (V2) at (-3, 3) {$v_2$};
\node [style=node,fill=red!30] (V6) at (3, 0) {$v_6$};
\node [style=node,fill=blue!30] (V7) at (6, 0) {$v_7$};
\node [style=node,fill=blue!30] (V4) at (-6, 0) {$v_4$};
\node [style=node,fill=blue!30] (V8) at (-3, -3) {$v_8$};
\end{tikzpicture}
\hspace{10mm}
\begin{tikzpicture}[scale=0.45]
\node at (-6,5) {$(d)$};
\node [style=node,fill=blue!30] (V9) at (3,-3) {$v_9$};
\node [style=node,fill=red!30] (V5) at (-3, 0) {$v_5$};
 \node [style=node,fill=red!30] (V3) at (3, 3) {$v_3$};
 \node [style=node,fill=blue!30] (V1) at (0, 5) {$v_1$};
 \node [style=node,fill=red!30] (V2) at (-3, 3) {$v_2$};
\node [style=node,fill=red!30] (V6) at (3, 0) {$v_6$};
\node [style=node,fill=blue!30] (V7) at (6, 0) {$v_7$};
\node [style=node,fill=blue!30] (V4) at (-6, 0) {$v_4$};
\node [style=node,fill=blue!30] (V8) at (-3, -3) {$v_8$};
    \draw [style=edge] (V5) -- (V4);
    \draw [style=edge] (V5) -- (V8);
    \draw [style=edge] (V6) -- (V9);
    \draw [style=edge] (V6) -- (V7);
    \draw [style=edge] (V3) -- (V6);
    \draw [style=edge] (V5) -- (V2);
    \draw [style=edge] (V3) -- (V1);
    \draw [style=edge] (V2) -- (V1);
    \draw [style=edge] (V4) to[out=100,in=170] (V1);
    \draw [style=edge] (V1) to[out=10,in=80] (V7);
    \draw [style=edge] (V5) -- (V6);
    \draw [style=edge] (V8) -- (V9);
\end{tikzpicture}
    \caption{In (a) every node is under majority illusion. An optimal solution to the MIAE problem, adding $11$ edges, is shown in (b); an optimal solution to the MIRE problem, requiring the removal of all $12$ edges, is shown in (c); an optimal solution to the MIE problem, adding $3$ edges and removing $3$ edges, is shown in (d).}
    \label{fig:complete-illusion}
\end{figure}

We remark that the assumption that blue party is strict majority is not necessary. However the assumption of strictness ensures that, for all three problems, it is always feasible to alter the graph to eliminate majority illusion. 

Grandi et al.~\cite{Grandi2023} left unresolved the computational complexity of these problems. Instead, rather than {\em completely} eliminating illusion at every node, they showed it was NP-complete to {\em partially} eliminate illusion for a fraction of the nodes.
In particular, they studied the following {\em $q$-partial (majority) illusion elimination problem}: can 
{\bf at most} a $q$-fraction of the nodes be left under illusion 
after altering (adding and/or deleting) at most $k$ edges in $G$.

They proved this problem is NP-complete for $q \in (0,1)$.   
Observe that the MIAE, MIRE, and MIE problems correspond to the
case $q=0$, that is, where no nodes are allowed to be left under illusion.
Therefore, the computational complexity of the main problems of interest remained open and motivated this work.


\subsection{Our Contributions.}

We make two contributions regarding illusion elimination in social networks. First, we prove that the majority illusion elimination problem is
solvable in polynomial time for all three variants (MIAE, MIRE and MIE).
Thus completely eliminating majority illusion in a social network is easy. The result is very surprising given, as described above, that the problem of partially eliminating illusion is hard~\cite{Grandi2023}.

In our second contribution, we generalize the problem
of eliminating majority illusion and consider the
problem of $p$-illusion.
The majority illusion elimination problems assume the blue party has at least $\frac12$ the votes. 
Suppose, instead, we do not assume the votes for the blue party and ask if it is possible to ensure that every voter believes that the blue party has at least a $p$-fraction of the ballots, where $p\in [0,1]$ is a rational number. Thus, we define the \emph{$p$-Illusion problem}: Given a graph $G= (V,E)$, a colouring function $f:V\rightarrow \{B, R\}$, and a rational number $p 
\in [0,1]$, can we modify at most $k$ edges such that {\em every} node has at least
$p$-fraction of blue nodes in its neighbourhood?
This question induces three problems for $p$-Illusion (analogous to majority illusion elimination), namely, $p$-IA, $p$-IR, and $p$-I,\footnote{We emphasize the distinction between the $p$-I problem and the
$q$-partial MIE problem or the $q$-majority illusion problem of~\cite{Grandi2023}. In $p$-I (and $p$-IA/$p$-IR)
{\bf every} node must have at least a $p$-fraction of blue neighbors.
In contrast, $q$-partial MIE eliminates the illusion of any {\bf $(1-q)$-fraction} of the nodes; and in the $q$-majority illusion problem asks if at least $q$ fraction of the nodes be under illusion.}
which as we prove are all polynomial time solvable for $p=\frac12$.
However, the majority case $p=\frac12$ is very special.
We prove that $p$-IA, $p$-IR, and $p$-I 
are all NP-complete for any rational $p\in [0,1]$ except for
$p\in \{0,\frac13, \frac12, \frac23, 1\}$.

An overview of the rest of the paper is as follows.
In Section~\ref{sec:related} we discuss related works.
In Section~\ref{sec:pre}, we present preliminary results and
observations concerning illusion elimination.
Our three polynomial time algorithms for MIAE, MIRE, and MIE
are then given in Sections~\ref{sec:AD}, \ref{sec:remove}, and  \ref{sec:both}, respectively.
Finally, in Section~\ref{sec:hardness} we prove the NP-completeness
of the $p$-IA, $p$-IR and $p$-I problems, for $p\notin \{0,\frac13, \frac12, \frac23, 1\}$.

\subsection{Related Work.}\label{sec:related}

Several studies have delved into the dynamics of decision-making within social networks and the manipulation of these networks, offering insights into various related phenomena.
The article 
by Grandi et al.~\cite{Grandi2023} is the most relevant, as it directly addresses the phenomenon of majority illusion in social networks, where individuals perceive a majority opinion that does not reflect the actual distribution of opinions. Through analysis and experimentation, the authors propose methods for identifying and mitigating majority illusion, offering valuable insights into the dynamics of opinion formation and social influence in networked environments. The results in this paper are based upon their model.

Liu et al. \cite{infogerry} explore the concept of ``information gerrymandering" in elections, where the strategic distribution of individuals within social networks can influence voting outcomes. They introduces a metric called {\em influence assortment} to quantify this phenomenon and demonstrates its effectiveness in predicting voting outcomes in various network structures.

The effect of adding and/or deleting edges in a social network
has been investigated in several settings. For example,
Castiglioni et al.~\cite{castiglioni2020election} study election control in social networks through the addition or removal of connections between individuals. Their study proposes methods for strategically manipulating network connections to influence election outcomes, highlighting the potential for network topology to shape democratic processes.

By analyzing the impact of link additions on information diffusion and exposure diversity, Becker et al.~\cite{becker2023improving} propose methods to mitigate bias and improve the accessibility of diverse viewpoints. Through empirical analysis and simulations, their study provides insights into the potential of structural interventions to promote fairness and inclusivity in online information dissemination.

Kempe et al.~\cite{kempe2020inducing} explore methods for inducing equilibria in networked public goods games by modifying the structure of the network. Through mathematical modeling and simulations, their research investigates strategies for promoting cooperation and achieving desirable outcomes in public goods dilemmas. Their work also proffers insights into the role of network topology in shaping collective behavior.

Wilder and Vorobeychik~\cite{wilder2018controlling} investigate opinion diffusion and campaigning strategies on society graphs. Their research studies how opinions spread through networks and the effectiveness of various campaigning approaches. They focus on the dynamics of opinion formation and diffusion in social networks, with the aim of understanding influence dynamics and designing effective communication strategies.

Faliszewski et al.~\cite{Faliszewski} examine techniques for controlling elections through social influence. By manipulating the spread of information and opinions within social networks, they consider strategies for influencing voter behavior and election outcomes and provides insights into the potential for social influence to shape electoral processes, highlighting the need for understanding and addressing such dynamics in democratic systems.

In a recent work \citet{fioravantes2024eliminating} study the problem of eliminating majority illusion by altering the color of at most $k$ vertices and show that, unlike our edge editing problem, it is NP-hard and W[2]-hard with respect to $k$. They present FPT algorithm parameterized by treewidth of the input graph and PTAS for planar graphs.

%% file: preliminaries.tex
\section{Preliminaries}\label{sec:pre}
We begin by defining the following notations.
Given a graph $G=(V,E)$, let $B$ be the set of blue nodes and $R$ the set of red nodes. For any node $v \in V$, let $\Gamma(v)$ denote the neighbourhood of $v$, $r(v) = |\{ u \in R : u \in \Gamma(v) \}|$ be the number of red nodes in the neighbourhood of $v$, and $b(v) = |\{ u \in B : u \in \Gamma(v) \}|$ be the number of blue nodes in its neighbourhood.

\begin{defn}[Illusion]
    Let $G=(V,E)$ be a graph such that $\frac{|B|}{|V|} > p$. Then, a node $v 
 \in V$ is under illusion if it has strictly less than $p$-fraction of blue neighbors, i.e., $\frac{b(v)}{|\Gamma(v)|} <  p$. When $p=1/2$, it is known as majority illusion.
\end{defn} 

Equivalently, if $\frac{|B|}{|V|} < p$, then a node $v 
 \in V$ is under illusion if it has at least $p$-fraction of blue neighbors, i.e., $\frac{b(v)}{|\Gamma(v)|} \geq  p$.

Next, we describe basic properties that optimal solutions to our three \textit{majority} illusion elimination problems must possess.
For a node $v \in V$, if $r(v)> b(v)$ then we say $v$ is in {\em deficit}; if $r(v)< b(v)$ then we say
$v$ is in {\em surplus}. 
So the task is to add/remove as few edges as possible so that no node is in deficit.

We begin with the trivial observation that it is never of benefit to add an edge between two red nodes.
\begin{observation}\label{noRRedges}
    In an optimal solution $E'$ to an instance $(V,E,f,k)$ of MIE or MIAE, there are no edges of $E'\setminus E$ between two red nodes.
\end{observation}
\begin{proof}
   Assume there exist two red nodes $u,v$ with $uv \in E'\setminus E$. Then $E'\setminus uv$ is also be a solution to the instance $(V,E,f,k)$, because removing $uv$ can only decrease the ratio of red nodes to blue nodes in both affected neighbourhoods. Thus $uv$ cannot be added in an optimal solution.
\end{proof}
Analogously, it is never of benefit to remove an edge between two blue nodes.
\begin{observation}\label{noBBedges}
    In an optimal solution $E'$ to an instance $(V,E,f,k)$ of MIE or MIRE, there are no edges of $E\setminus E'$ between two blue nodes.
\end{observation}

\begin{proof}
   Assume there exist two blue nodes $u,v$ with $uv \in E \setminus E'$. Then $E'\cup uv$ is also a solution to the instance $(V,E,f,k)$, because adding $uv$ can only decrease the ratio of red nodes to blue nodes in both affected neighbourhoods. Thus $uv$ cannot be removed in an optimal solution.
\end{proof}

We conclude this section with three more observations, one for each of our three problems. The first, for the addition variant, essentially states that for any red node, it is never beneficial to connect it to more new blue neighbours than is absolutely necessary. 

\begin{lem}\label{capRBedges}
    In an optimal solution $E'$ to $(V,E,f,k)$ of MIAE, there are exactly $\max(r(v)-b(v),0)$ edges of $E'\setminus E$ incident to each red node $v$.
\end{lem}

\begin{proof}
    Suppose there are fewer than $\max(r(v)-b(v),0)$ edges of $E'\setminus E$ incident to some red node~$v$. Then $N_{E'}[v]$ is majority red, contradicting the feasibility of $E'$. 

    On the other hand, suppose there are more than $\max(r(v)-b(v),0)$ edges of $E'\setminus E$ incident to some red node $v$. Then, by Lemma \ref{noRRedges}, there are more than $\max(r(v)-b(v),0)$ edges between $v$ and the set of blue nodes. For any blue node $u$ with $uv \in E'\setminus E$, observe that $E'\setminus uv$ must also be a solution to the instance $(V,E,f,k)$. This is because removing $uv$ does not cause there to be more red nodes than blue nodes in the neighbourhood of $v$, and it decreases the ratio of red nodes to blue nodes in the  neighbourhood of $u$. Thus $uv$ cannot be added in an optimal solution, a contradiction.
\end{proof}

Similarly, for the removal variant, for any blue node, it is never beneficial to disconnect it from more red neighbours than is absolutely necessary. 

\begin{lem}\label{capBRedges}
    In an optimal solution $E'$ to $(V,E,f,k)$ of MIRE, there are exactly $\max(r(v)-b(v),0)$ edges of $E\setminus E'$ incident to each blue node $v$.
\end{lem}

\begin{proof}
    Suppose there are fewer than $\max(r(v)-b(v),0)$ edges of $E\setminus E'$ incident to some blue node $v$. Then $N_{E'}[v]$ is majority red, contradicting the feasibility of $E'$. 
    
    On the other hand, suppose 
    there are more than $\max(r(v)-b(v),0)$ edges of $E\setminus E'$ incident to some blue node $v$. Then, by Lemma \ref{noBBedges}, there are more than $\max(r(v)-b(v),0)$ edges between $v$ and the set of red nodes. For some red node $u$ with $uv \in E\setminus E'$, observe $E'\cup uv$ must also be a solution to the instance $(V,E,f,k)$. This is because adding $uv$ does not cause there to be more red nodes than blue nodes in the neighbourhood of $v$, and it decreases the ratio of red nodes to blue nodes in the  neighbourhood of $u$. Thus $uv$ cannot be removed in an optimal solution.
\end{proof}

Finally, for the general majority illusion elimination problem, we have the following more substantial observation. 
\begin{lem}\label{FormofMIE}
    There exists an optimal solution $E'$ to an instance $(V,E,f,k)$ of MIE with no edges of $E'\setminus E$ incident to any red node, and no edges of $E\setminus E'$ incident to any blue node. 
\end{lem}

\begin{proof}
We provide a method of transforming any optimal solution $E'$ to an instance $(V,E,f,k)$ of MIE into an optimal solution of the desired form.

Suppose there is an edge of $E'\setminus E$ incident to some red node $v$. By \Cref{noRRedges}, this edge connects $v$ to some node $u \in B$. $E'$ is an optimal solution, meaning $N_{E'}(v)$ contains at least as many blue nodes as red nodes. Removing $uv$ from $E'$ would reduce the size of $|E\setminus E'|+|E'\setminus E|$, and result in one fewer blue node in $N_{E'}(v)$ and one fewer red node in $N_{E'}(u)$, while leaving all other neighbourhoods unchanged. By the optimality of $E'$, $E'\setminus \{uv\}$ is not a solution, therefore it must be true that removing a blue node from $N_{E'}(v)$ creates a majority red neighbourhood, meaning there are exactly as many red nodes in $N_{E'}(v)$ as blue nodes. Therefore, we will alter $E'$ by removing all edges of $E'\setminus E$ incident to $v$, and removing an equal number of edges of $E'$ between $v$ and red nodes. This operation, maintains the size of $|E\setminus E'|+|E'\setminus E|$, decreases the number of edges of $E'\setminus E$ incident to a red node, and increases the number of edges between red nodes in $E\setminus E'$ by the same amount. 

Now, suppose there is an edge of $E\setminus E'$ incident to a blue node $u$. Then, by \Cref{noBBedges}, this edge connects $u$ to some node $v \in R$. $E'$ is an optimal solution, meaning $N_{E'}(v)$ contains at least as many blue nodes as red nodes. Adding $uv$ from $E'$ would reduce the size of $|E\setminus E'|+|E'\setminus E|$, and result in one more blue node in $N_{E'}(v)$ and one more red node in $N_{E'}(u)$, while leaving all other neighbourhoods unchanged. By the optimality of $E'$, $E'\cup \{uv\}$ is not a solution, therefore it must be true that adding a red node to $N_{E'}(u)$ creates a majority red neighbourhood, meaning there are exactly as many red nodes in $N_{E'}(u)$ as blue nodes. Therefore, we will alter $E'$ by adding all edges of $E\setminus E'$ incident to $u$, and adding an equal number of edges of $E$ between $u$ and blue nodes. This operation, maintains the size of $|E\setminus E'|+|E'\setminus E|$, decreases the number of edges of $E\setminus E'$ incident to a blue node, and increases the number of edges between blue nodes in $E'\setminus E$ by the same amount. 

Repeated application of these two operations returns an optimal solution $E'$ with no edges of $E'\setminus E$ incident to any red node, and no edges of $E\setminus E'$ incident to any blue node.
\end{proof}

%% file: add.tex
\section{Majority Illusion Addition Elimination (MIAE)}\label{sec:AD}

The Majority Illusion Addition Elimination (MIAE) Problem is the Majority Illusion Elimination Problem (MIE) restricted to adding edges.

\subsection{A First Attempt}

Given $G=(B\cup R, E)$ we define an auxiliary graph $G^*=(B\cup R, E^*)$, called the {\em availability graph}. 
Motivated by Lemma~\ref{noRRedges}, its edges are the edges in the complement of $G$ excluding those between two red nodes. That is, $E^*$ = $\overline{E}(R,B) \cup \overline{E}(G[B])$. Now let 
      $$
   b^*(v) =
    \begin{cases}
\max(r(v)-b(v),0) &\mathrm{if}\ v \in R \\
b(v) - r(v) &\mathrm{if}\ v \in B
    \end{cases}
    $$

Now consider the following integer program.
\begin{align}
{\mathrm ({\tt IP-1a})}\qquad \qquad \max \quad \sum\limits_{e\in E^*(G^*)} -y_e  \label{bad:obj}\\
\text{s.t.} \ \hspace{3cm}\  \displaystyle\sum\limits_{
  u:u \in N_{E^*}(v)\cap B}   
  y_{uv}  &= b^*(v)  \qquad  v\in R(G^*) \label{bad:red}\\
  \displaystyle\sum\limits_{
  u:u \in N_{E^*}(v)\cap B}   
  -y_{uv} +
  \displaystyle\sum\limits_{
  u:u\in N_{E^*}(v) \cap R}    
  y_{uv} &\leq b^*(v)  \qquad v\in B(G^*) \label{bad:blue}\\
    y_e &\in \{0,1\} \qquad e\in E^*(G^*) \label{bad:01}
\end{align}

\begin{thm}\label{thm:badIP}
    The integer program~{\em ({\tt IP-1a})} is the MIAE problem.
\end{thm}
\begin{proof}
Observe, by (\ref{bad:01}),  that an optimal solution ${\bf y}$ to the integer program~{\em ({\tt IP-1a})} corresponds to a set of edges $Y$.
We will demonstrate that $Y$ is an optimal solution to the MIAE problem. 

First, by \Cref{noRRedges}, 
an optimal solution to the MIAE problem is a subset of $E^*$.
Thus constraint (\ref{bad:01}) is valid.

Second, by \Cref{capRBedges}, an optimal solution to the MIAE problem has the property that there are exactly 
$\max(r(v)-b(v),0)$ edges of incident to each red node $v$. Thus the equality constraints (\ref{bad:red}) are valid.

Third, any feasible solution to the MIAE problem must satisfy the property that the number of red nodes does not exceed the number of blue nodes in the neighbourhood of any blue node. That is, 
\begin{align}
  r(v) + \displaystyle\sum\limits_{
  u:u\in N_{E^*}(v) \cap R}    
  y_{uv} &\leq b(v) + \displaystyle\sum\limits_{
  u:u \in N_{E^*}(v)\cap B}   
  y_{uv}  \qquad v\in B(G^*) \label{bad:blue2}
\end{align}
But, for any blue node $v$, by definition $b^*(v)=b(v)-r(v)$. Thus rearranging (\ref{bad:blue2}) we obtain 
(\ref{bad:blue}).

Finally, an optimal solution to the MIAE problem contains as few edges of $E^*$ as possible.
However, minimizing $\sum_{e\in E^*(G^*)} y_e$ 
is equivalent to maximizing the objective function (\ref{bad:obj}).

It follows that an optimal solution ${\bf y}$ to integer program (IP-1a) is an optimal solution to the MIAE problem, as desired. 
\end{proof}

Therefore to solve the MIAE problem, we can simply solve the integer program ({\tt IP-1a}). But running an integer programming solver
is not polynomial time in the worst case, so we need a more efficient
method. Observe that the ({\tt IP-1a}) has size polynomial in the cardinality of the graph. Thus, its linear program relaxation can be solved in polynomial time. So, if the optimal solution to the LP relaxation is always integral then,
by Theorem~\ref{thm:badIP}, we would have an efficient algorithm for the MIAE problem. Unfortunately, this property is not guaranteed. The LP relaxation is not integral.

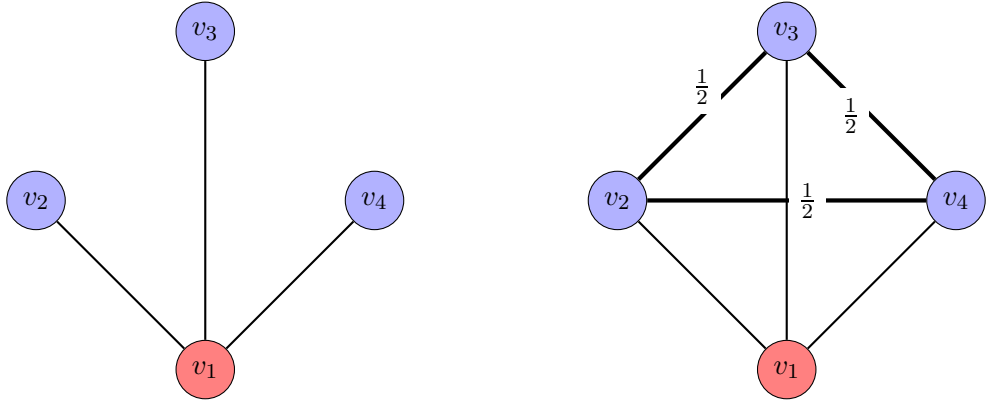
\begin{figure}[!ht]
    \centering
\tikzstyle{node}=[ draw=black,
 shape=circle, minimum size=10pt]
 \tikzstyle{edge}=[thick, -]


\begin{tikzpicture}[scale=0.75]
	\node [style=node,fill=red!50] (v1) at (0,-3) {$v_1$};
	\node [style=node,fill=blue!30] (v2) at (-3, 0) {$v_2$};
	\node [style=node,fill=blue!30] (v3) at (0, 3) {$v_3$};
    \node [style=node,fill=blue!30] (v4) at (3, 0) {$v_4$};
    \draw [style=edge] (v1) -- (v2);
    \draw [style=edge] (v1) -- (v3);
    \draw [style=edge] (v1) -- (v4);
\end{tikzpicture}
\qquad \qquad \qquad
\begin{tikzpicture}[scale=0.75]
	\node [style=node,fill=red!50] (v1) at (0,-3) {$v_1$};
	\node [style=node,fill=blue!30] (v2) at (-3, 0) {$v_2$};
	\node [style=node,fill=blue!30] (v3) at (0, 3) {$v_3$};
    \node [style=node,fill=blue!30] (v4) at (3, 0) {$v_4$};
    \draw [style=edge] (v1) -- (v2);
    \draw [style=edge] (v1) -- (v3);
    \draw [style=edge] (v1) -- (v4);
    \draw [style=edge, ultra thick] (v3) -- (v4) node [midway,above, left, fill=white] {$\frac12$};
    \draw [style=edge, ultra thick] (v2) -- (v4) node [midway,above,right, fill=white] {$\frac12$};
    \draw [style=edge, ultra thick] (v2) -- (v3) node [midway,above, fill=white] {$\frac12$};
\end{tikzpicture}
    \caption{A non-integral optimal solution to the linear program relaxation of ({\tt IP-1a}).}
    \label{fig:non-integral}
\end{figure}

To see this, consider the graph shown in the LHS of Figure~\ref{fig:non-integral}. Each of the three blue nodes suffers from majority illusion as it has only a red node
in its neighbour.
So they need an additional blue neighbour.
To rectify this the optimal solution ({\tt IP-1a}) requires the addition of two edges.
However, the linear program relaxation has optimal
value $\frac32$. This is illustrated in the RHS of Figure~\ref{fig:non-integral}; simply add a cycle on the
three blue nodes, shown in bold, and assign these new edges weight $\frac12$. 

\subsection{A Second Attempt}

So we desire an integer program that (i) exactly models the MIAE problem, (ii) has an integral LP relaxation, and (iii) the LP relaxation is solvable in polynomial time.

Towards this goal we make two alterations to the integer program ({\tt IP-1a}). The first is we alter the constraint (\ref{bad:blue}) for each blue node $v$ by adding $|N_{E^*}(v) \cap B|$ to both sides. This transformation is evidently innocuous. The new altered constraints are then
\begin{align}
    \displaystyle\sum\limits_{
  u:u \in N_{E^*}(v)\cap B}   
  (1-y_{uv}) +
  \displaystyle\sum\limits_{
  u:u\in N_{E^*}(v) \cap R}    
  y_{uv} &\leq b^*(v) + |N_{E^*}(v) \cap B| \qquad v\in B(G^*) 
\end{align}
For ease of exposition, we define a new function $b'(v)$ to account for this modification. Namely, let
$$
   b'(v) =
    \begin{cases}
\ \max(r(v)-b(v),0) &\mathrm{if}\ v \in R \\
\ |N_{E^*}(v) \cap B| + b(v) - r(v) &\mathrm{if}\ v \in B
    \end{cases}
    $$ 

Critically, we also add the following new set of constraints to the integer program.
\begin{equation}
\sum\limits_{
  \substack{
  e:e \in E^*(G^*[B])\\
  e \in E^*(G^*[X])\cup F}}
  (1-y_e)+ 
  \displaystyle\sum\limits_{
  \substack{
  e:e \in E^*(R,B)\\
  e \in E^*(G^*[X])\cup F}}
  y_e \ \leq \ 
  \Bigl\lfloor \frac12\Bigl(
\sum\limits_{
  v:v \in X}
  b'(v)+ |F| 
  \Bigr) \Bigr\rfloor  
\ \begin{array}{cc}
     & \forall X\subseteq B\cup R \\
     & \forall F \subseteq \delta_{E^*}(X)
\end{array}\label{new}
\end{equation}
We will show in Theorem~\ref{thm:IP} that the constraints are satisMIEd by any feasible solution to the MIAE problem.
For the motivation underlying these constraints consider
again the example in Figure~\ref{fig:non-integral}.

Consider the case of $X=B$ and $F=\emptyset$.
This gives a violated constraint for the non-integral solution ${\bf y}$ illustrated in Figure~\ref{fig:violated-constraint}.

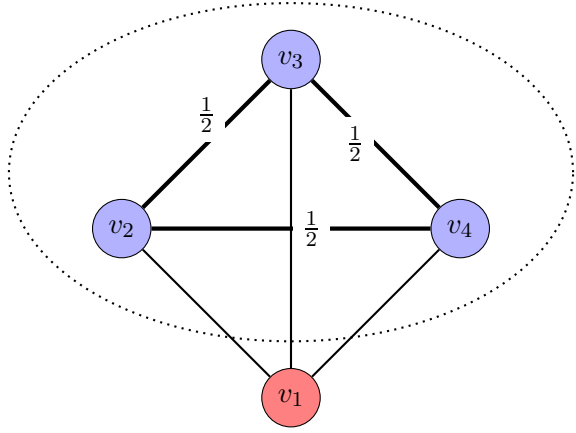
\begin{figure}[!ht]
    \centering
\tikzstyle{node}=[ draw=black,
 shape=circle, minimum size=10pt]
 \tikzstyle{edge}=[thick, -]
\begin{tikzpicture}[scale=0.75]
\draw [thick, dotted] (0,1) ellipse (5cm and 3cm);
	\node [style=node,fill=red!50] (v1) at (0,-3) {$v_1$};
	\node [style=node,fill=blue!30] (v2) at (-3, 0) {$v_2$};
	\node [style=node,fill=blue!30] (v3) at (0, 3) {$v_3$};
    \node [style=node,fill=blue!30] (v4) at (3, 0) {$v_4$};
    \draw [style=edge] (v1) -- (v2);
    \draw [style=edge] (v1) -- (v3);
    \draw [style=edge] (v1) -- (v4);
    \draw [style=edge, ultra thick] (v3) -- (v4) node [midway,above, left, fill=white] {$\frac12$};
    \draw [style=edge, ultra thick] (v2) -- (v4) node [midway,above,right, fill=white] {$\frac12$};
    \draw [style=edge, ultra thick] (v2) -- (v3) node [midway,above, fill=white] {$\frac12$};
\end{tikzpicture}
    \caption{A violated constraint: $1-y_{v_3v_4} + 1-y_{v_3v_2} + 1-y_{v_2v_4} \leq \lfloor \frac32 \rfloor$}
    \label{fig:violated-constraint}
\end{figure}
Observe, for this example, that (\ref{new}) becomes
\begin{equation}\label{new:B=x}
\sum\limits_{
  e:e \in E^*(G^*[B])}
  (1-y_e)
 \ \leq \ 
  \Bigl\lfloor \frac12\Bigl(
\sum\limits_{
  v:v \in B}
  b'(v) 
  \Bigr) \Bigr\rfloor  
\end{equation}
Now
\begin{align*}
\sum_{
  e:e \in E^*(G^*[B])}
  (1-y_e)
 &=  
  (1-y_{v_2v_3}) + (1-y_{v_2v_4}) + (1-y_{v_3v_4}) \\
  &= 3 - (y_{v_2v_3} + y_{v_2v_4} + y_{v_3v_4})\\
  &= 3- 3\cdot \frac12 \\
  &= \frac32
\end{align*}
On the other hand 
\begin{align*}
\Bigl\lfloor \frac12\Bigl(
\sum\limits_{
  v:v \in B}
  b'(v) 
  \Bigr) \Bigr\rfloor 
  &=
\Bigl\lfloor \frac12\Bigl(
\sum\limits_{
  v:v \in B}
\left( |N_{E^*}(v) \cap B| + b(v) - r(v) \right)
  \Bigr) \Bigr\rfloor \\
  &=
\Bigl\lfloor \frac12\Bigl(
\sum\limits_{
  v:v \in B}
\left( 2 + 0 - 1 \right)
  \Bigr) \Bigr\rfloor \\
  &=
\Bigl\lfloor \frac12
\sum\limits_{
  v:v \in B}
1
  \Bigr\rfloor \\
  &=\Bigl\lfloor \frac32 \Bigr\rfloor \\
  &= 1
\end{align*}
Ergo, the constraint (\ref{new:B=x}) is violated for the non-integral solution ${\bf y}$.
In fact, we will prove adding this set of new constraints will eliminate all non-integral solutions.

Putting all this together our new integer program is
\begin{align}
{\mathrm ({\tt IP-1b})}\qquad \max \quad \sum\limits_{e\in E^*(G^*)} -y_e  \nonumber\\
\text{s.t.} \ \hspace{2cm}\  \displaystyle\sum\limits_{
  u:u \in N_{E^*}(v)\cap B}   
  y_{uv}  &= b'(v)  \qquad  v\in R(G^*) \nonumber\\
  \displaystyle\sum\limits_{
  u:u \in N_{E^*}(v)\cap B}   
  (1-y_{uv}) +
  \displaystyle\sum\limits_{
  u:u\in N_{E^*}(v) \cap R}    
  y_{uv} &\leq b'(v)  \qquad v\in B(G^*) \nonumber\\
\sum\limits_{
  \substack{
  e:e \in E^*(G^*[B])\\
  e \in E^*(G^*[X])\cup F}}
  (1-y_e)+ 
  \displaystyle\sum\limits_{
  \substack{
  e:e \in E^*(R,B)\\
  e \in E^*(G^*[X])\cup F}}
  y_e &\leq 
  \Bigl\lfloor \frac12\Bigl(
\sum\limits_{
  v:v \in X}
  b'(v)+ |F| 
  \Bigr) \Bigr\rfloor  
\ \begin{array}{cc}
     & \forall X\subseteq B\cup R \\
     & \forall F \subseteq \delta_{E^*}(X)
\end{array} \nonumber\\  
    y_e &\in \{0,1\} \qquad e\in E^*(G^*) \nonumber
\end{align}

Our first step then is to verify that (i) still holds; that is, ({\tt IP-1b}) solves the MIAE problem.

\begin{thm}\label{thm:IP}
    The integer program~{\em ({\tt IP-1b})} is the MIAE problem.
\end{thm}
\begin{proof}
By Theorem~\ref{thm:badIP}, it suffices to prove that
the integer program~{\em ({\tt IP-1b})} is the integer program~{\em ({\tt IP-1a})}.
As mentioned, the alterations to the previous inequality constraints are trivial. What remains is to show the newly added constraints (\ref{new}) do not eliminate any solutions. So take any solution $\bf{y}$ to the integer program~{\em ({\tt IP-1a})}, any collection $X\subseteq B\cup R$ and any subset of edges $F\subseteq \delta_{E^*}(X)$ leaving $X$ in $E^*(G^*)$. For each node $v \in B$, the constraint on $v$ gives  
\begin{equation}
\sum\limits_{\substack{
  e: v \in e\\
  e \in E^*(G^*[B])\\
  e \in E^*(G^*[X])}} (1 -y_e) + \displaystyle\sum\limits_{\substack{
  e: v \in e\\
  e \in E^*(R,B)\\
  e \in E^*(G^*[X])}} y_e+
  \displaystyle\sum\limits_{\substack{
  e: v\in e\\
  e \in E^*(G^*[B])\\
  e \in F}} (1- y_e) + 
\displaystyle\sum\limits_{\substack{
  e: v \in e\\
  e \in E^*(R,B)\\
  e \in  F}} y_e 
  \ \leq\  b'(v)
\end{equation}
Also, for each node $v \in R$, the constraint on $v$ gives  
\begin{equation}
 \displaystyle\sum\limits_{\substack{
  e: v \in e\\
  e \in E^*(R,B)\\
  e \in E^*(G^*[X])}} y_e+
\displaystyle\sum\limits_{\substack{
  e: v \in e\\
  e \in E^*(R,B)\\
  e \in  F}} y_e 
  \ \leq\  b'(v)
\end{equation}
Therefore, when we sum these constraints for every node in $X$, we get
\begin{equation}\label{eq:sup-dem}
2\cdot\sum\limits_{\substack{
  e:e \in E^*(G^*[B])\\
  e \in E^*(G^*[X])}} (1 -y_e) + 2\cdot\displaystyle\sum\limits_{\substack{
  e:e \in E^*(R,B)\\
  e \in E^*(G^*[X])}} y_e+
  \displaystyle\sum\limits_{\substack{
  e:e \in E^*(G^*[B])\\
  e \in F}} (1- y_e) + 
\displaystyle\sum\limits_{\substack{
  e:e \in E^*(R,B)\\
  e \in  F}} y_e 
  \ \leq\ \displaystyle\sum\limits_{v:v \in X} b'(v)
\end{equation}

Also, trivially, we have
\begin{align}\label{eq:triv}
|F| &= 
\sum\limits_{
e:e \in E^*(G^*[B])\cap F} 1 + \sum\limits_{
e:e \in E^*(R,B)\cap F
} 1 \nonumber \\
&\ge
\sum\limits_{
e:e \in E^*(G^*[B])\cap F}(1-y_e)+ \sum\limits_{
e:e \in E^*(R,B)\cap F
}y_e 
\end{align}
Summing (\ref{eq:sup-dem}) and  (\ref{eq:triv}), we obtain

\begin{equation*}
2\cdot\sum\limits_{\substack{
e:e \in E^*(G^*[B])\\
e \in E^*(G^*[X])\cup F}}(1-y_e)+ 2\cdot\sum\limits_{\substack{
e:e \in E^*(R,B)\\
e \in E^*(G^*[X])\cup F}} y_e 
\ \leq\  \sum\limits_{v:v \in X}b'(v)+ |F|
\end{equation*}

Finally divide both sides by two and take their floors. By the integrality of ${\bf y}$, we have
\begin{equation*}
\sum\limits_{\substack{
e:e \in E^*(G^*[B])\\
e \in E^*(G^*[X])\cup F}}(1-y_e)+ \sum\limits_{\substack{
e:e \in E^*(R,B)\\
e \in E^*(G^*[X])\cup F }}y_e 
\ \leq\  
\Bigl\lfloor \frac12\Bigl(\sum\limits_{v:v \in X}b'(v)+ |F|\Bigr) \Bigr\rfloor
\end{equation*}
Thus ${\bf y}$ is also a solution to the integer program~{\em ({\tt IP-1b})}.
\end{proof}

Again, Theorem~\ref{thm:IP} implies that we can solve the MIAE problem by solving the integer program ({\tt IP-1b}). 
We will do this efficiently by solving its linear program relaxation.
To achieve this successfully, we must show that the LP relaxation is
integral and also that it can be solved in polynomial time despite the fact that there are an exponential number of constraints of type (\ref{new}).

\subsection{A TDI System}
The LP relaxation of ({\tt IP-1b}) is
 \begin{align}
 {\mathrm ({\tt IP-1b})}\qquad \max \quad\sum\limits_{e\in E^*(G^*)} -y_e   \\
\text{s.t.} \ \hspace{2cm}\  \displaystyle\sum\limits_{
  u:u \in N_{E^*}(v)\cap B}   
  y_{uv}  &= b'(v)  \qquad  v\in R(G^*) \\
  \displaystyle\sum\limits_{
  u:u \in N_{E^*}(v)\cap B}   
  (1-y_{uv}) +
  \displaystyle\sum\limits_{
  u:u\in N_{E^*}(v) \cap R}    
  y_{uv} &\leq b'(v)  \qquad v\in B(G^*) \\
\sum\limits_{
  \substack{
  e:e \in E^*(G^*[B])\\
  e \in E^*(G^*[X])\cup F}}
  (1-y_e)+ 
  \displaystyle\sum\limits_{
  \substack{
  e:e \in E^*(R,B)\\
  e \in E^*(G^*[X])\cup F}}
  y_e &\leq 
  \Bigl\lfloor \frac12\Bigl(
\sum\limits_{
  v:v \in X}
  b'(v)+ |F| 
  \Bigr) \Bigr\rfloor  
\ \begin{array}{cc}
     & \forall X\subseteq B\cup R \\
     & \forall F \subseteq \delta_{E^*}(X)
\end{array} \\  
    0 \ \le \ y_e &\le 1 \qquad e\in E^*(G^*) 
\end{align}

It will now be cleaner to impose a change of variables.
Specifically, let
$$
    x_e =
    \begin{cases}
    y_e &\mathrm{if}\ e \in E^*(R,B)\\
    1-y_e &\mathrm{if}\ e \in E^*(G^*[B])
    \end{cases}
$$\label{cov}
The LP relaxation then simplifies to
\begin{align}
{\mathrm ({\tt LP-1c})}\qquad \max \quad \sum\limits_{e\in E^*(G[B])}x_e &- \sum\limits_{e\in E^*(R,B)}x_e - |E^*(G[B])| \\
\text{s.t.} \ \hspace{2cm}\  \sum\limits_{
  e:e \in\delta(v)}   
  x_e  &= b'(v)  \qquad  v\in R \\
  \sum\limits_{
  e:e \in\delta(v)}   
  x_e  &\le b'(v)  \qquad  v\in B \\
\sum\limits_{e:e \in E^*(G^*[X])} x_e  +\sum\limits_{e:e \in F} x_e &\leq \Bigl\lfloor \frac12\Bigl(\sum\limits_{v:v \in X}b'(v)+ |F|\Bigr) \Bigr\rfloor \qquad   
\begin{array}{l}
\forall X\subseteq B\cup R\\
\forall F \subseteq \delta_{E^*}(X)
\end{array} \\
    0 \ \le\  x_e &\le 1 \qquad e\in E^*(G^*) 
\end{align}

One method to prove that a polytope has integer nodes involves the concept of total dual integrality (TDI).
\begin{defn}
    A polytope $\textbf{Ax} \leq \textbf{b}$ is \textbf{Total Dual Integral (TDI)} if $\textbf{A,b}$ have rational entries, and if for any $\textbf{c} \in \mathcal{Z}^n$ such that  $\max~ \textbf{cx}$ s.t $\textbf{Ax} \leq \textbf{b}$ has an optimal solution, then there is an integer optimal dual solution.
\end{defn}

A proof that a linear program is TDI allows us to use the following result to show that it is integral.

\begin{thm}\label{TDIint}\cite{EDMONDS1977185}
    If a linear system $\{\textbf{Ax}\leq \textbf{b}\}$ is TDI, and $\textbf{b}$ is integral, then $\{\textbf{Ax}\leq \textbf{b}\}$ is integral.
\end{thm}

We will prove that ({\tt LP-1c}) is total dual integral (TDI) to show its integrality.
To do this, we turn our attention to simple $b$-matchings.

\begin{thm}\cite{Schrijver}\label{bTDI}
The following simple $\textbf{b'}$-matching polytope is TDI:
\begin{align}
({\tt \textbf{b'}-matching~LP})\qquad \sum\limits_{
  e:e \in\delta(v)}   
  x_e  &\le b'(v)  \qquad  v\in B\cup R \\
\sum\limits_{e:e \in E^*(G^*[X])} x_e  +\sum\limits_{e:e \in F} x_e &\leq \Bigl\lfloor \frac12\Bigl(\sum\limits_{v:v \in X}b'(v)+ |F|\Bigr) \Bigr\rfloor \qquad   
\begin{array}{l}
\forall X\subseteq B\cup R\\
\forall F \subseteq \delta_{E^*}(X)
\end{array} \\
    0 \ \le\  x_e &\le 1 \qquad e\in E^*(G^*)  
\end{align}
\end{thm} 

This TDI system is very similar to our system. We merely require one more result.

\begin{thm}\label{TDIpres}\cite{COOK198331}
Let $\textbf{Ax}\leq \textbf{b}$ be TDI and let $\textbf{A'x}\leq \textbf{b'}$ arise from $\textbf{Ax}\leq \textbf{b}$ by adding $-a^Tx\leq-\beta$ for some inequality $a^Tx \leq\beta$ in $\textbf{Ax}\leq \textbf{b}$. Then also $\textbf{A'x}\leq \textbf{b'}$ is TDI.
\end{thm}

\begin{cor}\label{cor:TDI}
    The linear system of (LP-1b) is TDI.
\end{cor}
\begin{proof}
The $\textbf{b'}$-matching LP is total dual integral
by \Cref{bTDI}. Therefore, the linear system of ({\tt LP-1c}) is total dual integral, as it is the result of applying \Cref{TDIpres} to the constraints of the $\textbf{b'}$-matching LP.
But, via the given change of variables, this implies the linear system of ({\tt LP-1b}) is total dual integral.
\end{proof}

\begin{cor}\label{cor:addint}
    The linear program {\em ({\tt LP-1b})} is integral.
\end{cor}
\begin{proof}
    This follows immediately from Corollary~\ref{cor:TDI} and Theorem~\ref{TDIint}.
\end{proof}

\subsection{A Polynomial Time Algorithm}
So ({\tt LP-1b}) is integral. But because it has exponentially many constraints, we need to show it can be solved in polynomial time. To do so, we apply a separation oracle.

\begin{defn}\label{sep}
    {\sc Separation Problem}: Given a convex set $P \subseteq \mathbb{R}^n$ and a vector $y \in \mathbb{Q}^n$, either determine that $y \in P$ or find a vector $d \in \mathbb{Q}^n$ such that $d^Tx < d^Ty$ for all $x \in P$.
\end{defn}

\begin{thm}\label{seplin}\cite{grötschel_lovász_schrijver_1988}
Let $n\in \mathbb{N}$ and $c \in \mathbb{Q}^n$. Let $P \subseteq \mathbb{R}^n$ be a rational polytope, and let $x_0\in \mathbb{Q}^n$ be a point in the interior of $P$. Let $T \in \mathbb{N}$ such that $size(x_0) \leq \log T$ and $size(x) \leq \log T$ for all nodes $x$ of $P$. Given $n,c,x_0, T$ and a polynomial-time oracle for the {\sc Separation Problem} for $P$, a node $x^*$ of $P$ attaining $max\{c^Tx: x\in P\}$ can be found in time polynomial in $n$, $\log T$ and $size(c)$.
\end{thm}

\begin{thm}\label{sepmatch}\cite{RP1981,LGR2008} For undirected graphs $G, \textbf{u}: E(G) \rightarrow \mathbb{N} \cup \{\infty\}$ and $\textbf{b}:V(G) \rightarrow \mathbb{N}$, the {\sc Separation Problem} for the $\textbf{b}$-matching polytope of $(G,\textbf{u})$ can be solved in $O(n^4)$ time.
\end{thm}

\begin{thm}\label{addn4}
    The {\sc Separation Problem} for {\em ({\tt LP-1b})} can be solved in $O(n^4)$ time.
\end{thm}

\begin{proof}
By Theorem \ref{seplin} it suffices to have a polynomial-time algorithm for the {\sc Separation Problem} for the vector ${\bf y}$ and the feasible set in ({\tt LP-1b}), which we will refer to as $P$. Such an algorithm exists for the convex set in ({\tt b’-matching LP}), which we will refer to as $Q$, by Theorem \ref{sepmatch}. If we define vector ${\bf x}\in \mathbb{R}^{E^*}$ by setting 
$$
    x_e =
    \begin{cases}
    y_e &\mathrm{if}\ e \in E^*(R,B)\\
    1-y_e &\mathrm{if}\ e \in E^*(G^*[B])
    \end{cases}$$ 
    and then change variables from ${\bf y}$ to ${\bf x}$, we can transform $P$ into a subset of $Q$. This means if ${\bf x} \not\in Q$, then ${\bf y} \not\in P$, and we can determine this in $O(n^4)$ time. We can also determine if ${\bf x} \in Q$ in $O(n^4)$ time, and if this is the case, we just need to check $|R|$ equality constraints to determine if ${\bf y} \in P$, which we can do in linear time.
\end{proof}

\begin{cor}\label{12IA}
    The {\sc $\frac12$ -Illusion Addition Problem} can be solved in polynomial time.
\end{cor}

\begin{proof}
    LP-1c will solve the $\frac12$-Illusion Addition problem just as it solves MIAE. When a strict majority of the vertices are blue, the two problems are equivalent. The removal of the strict majority requirement, which is the only diffence between MIAE and $\frac12$-IA, changes nothing except that it allows input graphs where LP - 1c has no solution, for example in cases where there is a blue vertex adjacent to more red vertices than there are blue vertices. In these cases, there is no solution of size less than or equal to $k$, meaning $\frac12$-IA is a NO instance.
\end{proof}

%% file: sub.tex
\section{Majority Illusion Removal Elimination (MIRE)}\label{sec:remove}

The Majority Illusion Removal Elimination (MIRE) Problem is the Majority Illusion Elimination (MIE) Problem restricted to removing edges


Given $G=(B\cup R, E)$ we define an auxiliary graph $G^*=(B\cup R, E^*)$, called the {\em removable graph}. Its edges are the edges in $G$ excluding those between two blue nodes. That is, $E^*=E(R,B) \cup E(G[R])$. Now let 
      $$
   b^*(v) =
    \begin{cases}
\max(r(v)-b(v),0) &\mathrm{if}\ v \in B \\
b(v) - r(v) &\mathrm{if}\ v \in R
    \end{cases}
    $$
Now consider the following integer program.
\begin{align}
{\mathrm ({\tt IP-2})}\qquad \max \quad \sum\limits_{e\in E^*(G^*)} -y_e  \label{badsub:obj}\\
\text{s.t.} \ \hspace{3cm}\  \displaystyle\sum\limits_{
  u:u \in N_{E^*}(v)\cap R}   
  y_{uv}  &= b^*(v)  \qquad  v\in B(G^*) \label{bad:subblue}\\
  \displaystyle\sum\limits_{
  u:u \in N_{E^*}(v)\cap R}   
  -y_{uv} +
  \displaystyle\sum\limits_{
  u:u\in N_{E^*}(v) \cap B}    
  y_{uv} &\leq b^*(v)  \qquad v\in R(G^*) \label{bad:subred}\\
    y_e &\in& \{0,1\} \qquad e\in E^*(G^*) \label{badsub:01}
\end{align}

\begin{thm}
    The integer program {\em ({\tt IP-2})} is the MIRE problem
\end{thm}
\begin{proof}
The proof is the same as for the MIAE case, as (IP-2) is symmetric to (IP-1) with the roles of $B$ and $R$ swapped. 
\end{proof}

As one might expect, our strategy for solving the MIRE problem will mirror the approach for the MIAE problem. We will modify ({\tt IP-2}) to create a linear program ({\tt LP-2}) which is symmetric to ({\tt LP-1b}).

Just like for ({\tt LP-1b}), we add a value to both sides of each inequality. In this case, it is $|N_{E^*}(v)\cap R|$. Once again we define function $b'(v)$ to account for this change:
$$
   b'(v) =
    \begin{cases}
\max(r(v)-b(v),0) &\mathrm{if}\ v \in B \\
|N_{E^*}(v) \cap R| + b(v) - r(v) &\mathrm{if}\ v \in R
    \end{cases}
    $$ 

Finally, we add an analogous set of new constraints to the integer program.
\begin{equation}
\sum\limits_{
  \substack{
  e:e \in E^*(G^*[R])\\
  e \in E^*(G^*[X])\cup F}}
  (1-y_e)+ 
  \displaystyle\sum\limits_{
  \substack{
  e:e \in E^*(R,B)\\
  e \in E^*(G^*[X])\cup F}}
  y_e \ \leq \ 
  \Bigl\lfloor \frac12\Bigl(
\sum\limits_{
  v:v \in X}
  b'(v)+ |F| 
  \Bigr) \Bigr\rfloor  
\ \begin{array}{cc}
     & \forall X\subseteq B\cup R \\
     & \forall F \subseteq \delta_{E^*}(X)
\end{array}
\end{equation}

The LP relaxation of ({\tt IP-2}) is
 \begin{align}
{\mathrm ({\tt LP-2})}\qquad \max \quad \sum\limits_{e\in E^*(G^*)} -y_e  \\
\text{s.t.} \ \hspace{2cm}\  \displaystyle\sum\limits_{
  u:u \in N_{E^*}(v)\cap R}   
  y_{uv}  &= b'(v)  \qquad  v\in B(G^*) \\
  \displaystyle\sum\limits_{
  u:u \in N_{E^*}(v)\cap R}   
  (1-y_{uv}) +
  \displaystyle\sum\limits_{
  u:u\in N_{E^*}(v) \cap B}    
  y_{uv} &\leq b'(v)  \qquad v\in R(G^*) \\
\sum\limits_{
  \substack{
  e:e \in E^*(G^*[R])\\
  e \in E^*(G^*[X])\cup F}}
  (1-y_e)+ 
  \displaystyle\sum\limits_{
  \substack{
  e:e \in E^*(R,B)\\
  e \in E^*(G^*[X])\cup F}}
  y_e &\leq 
  \Bigl\lfloor \frac12\Bigl(
\sum\limits_{
  v:v \in X}
  b'(v)+ |F| 
  \Bigr) \Bigr\rfloor  
\ \begin{array}{cc}
     & \forall X\subseteq B\cup R \\
     & \forall F \subseteq \delta_{E^*}(X)
\end{array} \\  
    0 \ \le \ y_e &\le 1 \qquad e\in E^*(G^*) 
\end{align}

\begin{cor}
    The linear program {\em ({\tt LP-2})} is integral.
\end{cor}

\begin{proof}
    The proof is identical to the proof of \Cref{cor:addint}.
\end{proof}

\begin{cor}
    The {\sc Separation Problem} for linear program {\em ({\tt LP-2})} can be solved in $O(n^4)$ time.
\end{cor}

\begin{proof}
    The proof is identical to the proof of \Cref{addn4}.
\end{proof}

\begin{cor}\label{12IR}
    The {\sc $\frac12$ -Illusion Removal Problem} can be solved in polynomial time.
\end{cor}

\begin{proof}
    LP-2 will solve the $\frac12$-Illusion Removal problem just as it solves MIRE. When a strict majority of the vertices are blue, the two problems are equivalent. The removal of the strict majority requirement, which is the only difference between MIRE and $\frac12$-IR, changes nothing.
\end{proof}

%% file: both.tex
\section{Majority Illusion Elimination (MIE)}\label{sec:both}

Finally, we arrive at the the general problem. Take $G=(B\cup R,E)$ and let $N_E(v)=\{u: (u,v)\in E\}$ be the neighbourhood of $v$. 
To solve the general problem, we will create new graphs $G^a= (R \cup B, E^a)$ and $G^s= (R \cup B, E^s)$.
 Specifically:
 
 $\bullet$ Create $G^a$ from $G$ by removing all edges of $G[R]$.
 
  $\bullet$ Create $G^r$ from $G$ by adding all edges of the complement of $G[B]$.

The remarkable fact now is that an optimal solution to the general problem on $G=(V,E)$ can be reduced to an optimal solution to the MIAE problem on $G^a$ plus an optimal solution to the MIRE problem on $G^s$.

\begin{thm}\label{thm:both}
    There is an optimal solution $E^*$ to the MIE problem on $G$ composed of the union of the optimal solution $E^a$ to the MIAE problem on $G^a$ and the optimal solution $E^r$ to the MIRE problem on $G^r$
\end{thm}

\begin{proof}
     By \Cref{noBBedges} and \Cref{FormofMIE}, the edges removed from $G$ to create $G^a$ are a superset of those removed in an optimal solution $E^*$ to MIE on $G$. These removals are strictly beneficial, as all edges removed are between two red nodes. Therefore any optimal solution $E^a$ to $MIAE$ on $G^a$ will be a subset of the edges added in some optimal solution $E^*$ to MIE on $G$. 

     By \Cref{noRRedges} and \Cref{FormofMIE}, the edges added to $G$ to create $G^r$ are a superset of those added in an optimal solution $E^*$ to MIE on $G$. These additions are strictly beneficial, as all edges added are between two blue nodes.  Therefore any optimal solution $E^r$ to $MIRE$ of $G^r$ will be a subset of the edges removed in some optimal solution $E^*$ to MIE on $G'$.

     Because $|E^r|+ |E^a| \leq |E^*|$, all that remains is to show that $E^r$ together with $E^a$ eliminates illusion in $G$. 
     Observe that in $G^r$ the blue nodes from a clique.
    Consequently, no blue node is under illusion, and 
     only red nodes can be under illusion.
     Therefore the only edges removed in an optimal solution $E^r$ to the MIRE problem will be between two red nodes. That is, removing an edge between a red and blue node cannot help the red node and is not necessary for the blue node, as it is not under illusion.
     
     Similarly, in $G^a$ the red nodes from an independent set. Consequently, no red node is under illusion, and only blue nodes can be under illusion. Therefore the only edges added in an optimal solution $E^a$ to the MIAE problem will be between two blue nodes. That is, adding an edge between a red and blue node cannot help the blue node and is not necessary for the red node, as it is not under illusion.
          
     Furthermore, all the red nodes in $G^r$ are under exactly the same illusion as the red nodes in $G$; therefore $E^r$ eliminates all illusion on $R$. Symmetrically, all the blue nodes in $G^a$ are under exactly the same illusion as the blue nodes in $G$; therefore $E^a$ eliminates all illusion on $B$. Thus, crucially, $E^r$ and $E^a$ do not affect each other, because they do not add or remove any edges in $(R,B)$. Therefore together they eliminate all illusion in $G$. The theorem follows.
\end{proof}

Theorem~\ref{thm:both} implies we can solve the MIE problem in polynomial time. We reduce it to a MIAE probelm and a MIRE problem, which we can both solve in polynomial time by the methods developed in Sections~\ref{sec:AD} and ~\ref{sec:remove}.

However, we can exploit Theorem~\ref{thm:both}  to solve the FIA directly using a single linear program. This is because the reduction worked by essentially reducing the problem to to two independent $b$-matching problems within the same graph! 
In particular, take $G = (B \cup R, E)$ and define functions
\begin{align*}
b^a(v) &=
|N_{E^*}(v) \cap B| + b(v) - r(v) \qquad \forall v \in B\\
   b^r(v) &= |N_{E^*}(v) \cap R| + b(v) - r(v) \qquad \forall v \in R
\end{align*}
Our linear program is then
\begin{align}
{\mathrm ({\tt LP-3})}\qquad \max \quad \sum\limits_{e\in E^*(G^*)} -y_e  \\
\text{s.t.} \ \hspace{2cm}\  
  \displaystyle\sum\limits_{
  u:u \in N_{E^*}(v)\cap B}   
  (1-y_{uv}) &\leq b^a(v)  \qquad v\in B(G^*) \\\label{lp3c1}
\sum\limits_{
  \substack{
  e:e \in E^*(G^*[B])\\
  e \in E^*(G^*[X])\cup F}}
  (1-y_e) &\leq 
  \Bigl\lfloor \frac12\Bigl(
\sum\limits_{
  v:v \in X}
  b^a(v)+ |F| 
  \Bigr) \Bigr\rfloor  
\ \begin{array}{cc}
     & \forall X\subseteq B \\
     & \forall F \subseteq \delta_{E^*}(X)
\end{array} \\\label{lp3c2}
\displaystyle\sum\limits_{
  u:u \in N_{E^*}(v)\cap R}   
  (1-y_{uv})  &\leq b^r(v)  \qquad v\in R(G^*) \\\label{lp3c3}
\sum\limits_{
  \substack{
  e:e \in E^*(G^*[R])\\
  e \in E^*(G^*[X])\cup F}}
  (1-y_e) &\leq 
  \Bigl\lfloor \frac12\Bigl(
\sum\limits_{
  v:v \in X}
  b^r(v)+ |F| 
  \Bigr) \Bigr\rfloor  
\ \begin{array}{cc}
     & \forall X\subseteq  R \\
     & \forall F \subseteq \delta_{E^*}(X)
\end{array} \\ \label{lp3c4}
    0 \ \le \ y_e &\le 1 \qquad e\in E^*(G^*) 
\end{align}

\begin{cor}
    The linear program {\em {\tt (LP-3)}} is integral.
\end{cor}

\begin{proof}
    This linear program is equivalent to two independent copies of the $b'$-matching LP, therefore its integrality follows directly from \Cref{bTDI} and \Cref{TDIpres}.
\end{proof}

\begin{cor}\label{subn4}
    The {\sc Separation Problem} for linear program {\em {\tt (LP-3)}} can be solved in $O(n^4)$ time.
\end{cor}

\begin{proof}
    This follows directly from \Cref{sepmatch}.
\end{proof}

\begin{cor}\label{12I}
    The {\sc $\frac12$-Illusion Problem} can be solved in polynomial time.
\end{cor}

\begin{proof}
    From the proof of \Cref{FormofMIE}, we know it is always preferable to correct a blue node in deficit by adding a blue neighbour than removing a red neighbour. However, without the strict majority restriction, this is not always possible. The fix is rather simple. For the set of blue nodes $B' = \{v: v \in B, N_R(v) \geq |B|\}$, we add edges between every node of $B'$ to all nodes of $B$, including $B'$, as we know there is a solution that contains all these edges. Now the  remaining edges of the solution to $\frac12$-Illusion consist of optimally removing edges of $G[R \cup B']$ and optimally adding non-edges of $G[B\setminus B']$, both of which we can do in polynomial time by \Cref{12IA} and \Cref{12IR}.
\end{proof}

%% file: hardness.tex
\section{Hardness}\label{sec:hardness}
In this section we present variants of the MIE, MIAE and MIRE problems. 

\begin{defn} \label{def:pMIAE}
    Let \(p \in [0,1], p \in \mathcal{Q}\). We define the \(p\)-IA problem as follows: Given a graph \(G = (V = R \cup B, E)\) and integer $k$, the objective is to determine if the addition of at most $k$ edges can ensure each node has at least a \(p\)-fraction of blue nodes in its neighbourhood.

\end{defn}

\begin{defn} \label{def:pMIRE}
    Let \(p \in [0,1], p \in \mathcal{Q}\). We define the \(p\)-IR problem as follows: Given a graph \(G = (V = R \cup B, E)\) and integer $k$, the objective is to determine if the subtraction of at most $k$ edges can ensure each node has at least a \(p\)-fraction of blue nodes in its neighbourhood.

\end{defn}

\begin{defn} \label{def:pMIE}
    Let \(p \in [0,1], p \in \mathcal{Q}\). We define the \(p\)-I problem as follows: Given a graph \(G = (V = R \cup B, E)\) and integer $k$, the objective is to determine if the addition/subtraction of at most $k$ edges can ensure each node has at least a \(p\)-fraction of blue nodes in its neighbourhood.

\end{defn}

For \(p = \frac{1}{2}\), if we add the requirement that blue be the strict majority, these extended problems correspond to the original MIAE, MIRE and MIE problems for which we have provided polynomial-time algorithms in the previous sections of this paper. Our algorithms can be extended to handle the case when blue is not the true majority, showing that the problems are solvable in polynomial time  when \(p = \frac{1}{2}\). 
Regardless, to show NP-hardness for other values of $p$, we provide a polynomial time reduction from the following NP-complete, satisfiability problem. 

\begin{defn}
    XSAT for $l$-CNF$_+^l$ is the problem of determining whether a boolean formula in conjunctive normal form composed of $n$ variables appearing exactly $l$ times as positive literals within clauses of size exactly $l$ has an assignment which evaluates to TRUE such that each clause has exactly one variable assigned TRUE.
\end{defn}

\begin{thm}\cite{de-2555}[Theorem 29]
    XSAT for $l$-CNF$_+^l$ is NP-complete for any $l \geq 3$
\end{thm}

\begin{thm}\label{thm:MIEp<12}
    The $p$-IA problem and $p$-I problem are NP-hard for $p\in \mathbb{Q}, 0 < p < \frac12, p \not=\frac13$.
\end{thm}

\begin{proof}

      First, we write $p$ as an irreductible fraction $p = \frac{a}{b}$ where $a, b \in \mathbb{N}$ are coprime and $p < \frac{1}{2}$. This implies $2a < b$. Therefore, we can write $p = \frac{a}{b} = \frac{a}{a + l}$ for $l \geq 3$, $l > a$, and $\gcd(a, l) = 1$. Note that $\frac{1}{3}$ cannot be written this way because of the condition $l \geq 3$.

    Consider an instance of XSAT for $l$-CNF$_+^l$ with $n$ variables and $n$ clauses, where $n \geq 2l$. We build an auxiliary graph $G = (B \cup R, E)$ of polynomial size such that there is a solution to this XSAT instance for $l$-CNF$_+^l$ if and only if the $p$-I problem associated with $G$ can be solved by adding or removing at most $\frac{n}{l} \cdot (l + a)$ edges. Moreover if there is a solution to this XSAT instance, there is a solution to the $p$-I problem that add exactly $\frac{n}{l} \cdot (l + a)$ edges and remove no edges. Consequently, this also works as a reduction to the $p$-IA problem. Note that $\frac{n}{l}$ must be an integer, as otherwise the XSAT problem trivially has no solution. 
    
    We design the graph $G$ to have nodes of six types: dummy blue nodes $B_D$, where $|B_D| = a \cdot n^2 + l - (2a \cdot n + 2)$; equalizing blue nodes $B_E$, where $|B_E| = 2a \cdot n - (n - 1)$; dummy red nodes $R_D$, where $|R_D| = l \cdot n^2 - 2l \cdot n + (n - l)$; equalizing red nodes $R_E$, where $|R_E| = 2l \cdot n - (n - l)$; blue variable nodes $B_V$, where $|B_V| = n$ (we have a blue variable node $v_x$ for each variable $x$ of the XSAT instance); and red clause nodes $R_C$, where $|R_C| = n$ (we have a red clause node $w_c$ for each clause $c$ of the XSAT instance).

    We now describe the graph $G$ by detailing the neighborhood of each node. We remind the reader of the notation $b(v)$ and $r(v)$, which denote the number of blue and red neighbors of node $v$, respectively.

    \begin{itemize}
    \item {\tt Variable nodes}: The neighbourhood of $v_x$ for all $v_x \in B_V$ is $B_E \cup R_E \cup \{w_c: w_c \in R_C, x \not\in c \text{ in the XSAT problem} \} \cup (B_V -v_x)$. Note that $b(v_x)=\frac{a}{l}\cdot r(v_x)$ for all $v_x \in B_V$, and therefore \emph{exactly a $p$-fraction} of the neighbourhood of each node of $B_V$ is blue.
    \item {\tt Clause nodes}: The neighbourhood of $w_c$ for all $w_c \in R_C$ is $B_E \cup R_E \cup B_D \cup R_D \cup \{v_x: v_x \in B_V, x \not\in C \text{ in the XSAT problem} \}$. Note that   $b(w_c) = \frac{a}{l}\cdot r(w_c)-1$ for all $w_c \in R_C$, so each $w_c \in R_c$ can add one more blue neighbour without violating the $p$-fraction constraint. 
    
    \item {\tt Dummy blue nodes}: The neighbourhood of $v$ for all $v \in B_D$ is $B_E\cup R_C \cup B_D - v$. There are  $an^2+l-n+2$ blue nodes and n red nodes in the neighbourhood of each $v \in B_D$, so there are at least $p$-fraction of blue nodes in the neighbourhood.
    \item {\tt Dummy red nodes}: The neighbourhood of $v$ for all $v \in R_D$ is $B_E \cup R_C$. There are $2an-(n-1)$ blue nodes and $n$ red nodes in the neighbourhood of each $v\in R_D$, so there are at least $p$-fraction of blue nodes in the neighbourhood.
    \item {\tt Equalizing blue nodes}: The neighbourhood of $v$ for all $v \in B_E$ is $R_C \cup R_D \cup B_V \cup B_D \cup B_E - v$. There are $an^2+l-2$ blue nodes and $l\cdot n^2+2n-2l\cdot n-l$ red nodes in the neighbourhood of each $v \in B_E$, so there are at least $p$-fraction of blue nodes in the neighbourhood.
    \item {\tt Equalizing red nodes}: The neighbourhood of $v$ for all $v \in R_E$ is $R_C \cup B_V$. There are at least $p$-fraction of blue nodes in the neighbourhood.
\end{itemize}

\begin{figure}[!ht]
    \centering
\tikzstyle{node}=[ draw=black,
 shape=circle, minimum size=10pt]
 \tikzstyle{edge}=[thick, -]


\begin{tikzpicture}[scale=0.5,x=3cm,y=3cm]
\node [style=node,fill=blue!30] (x1) at (6, 0) {$v_{x_6}$};
\node [style=node,fill=blue!30] (x2) at (5, 0) {$v_{x_5}$};
\node [style=node,fill=blue!30] (x3) at (4, 0) {$v_{x_4}$};
\node [style=node,fill=blue!30] (x4) at (3, 0) {$v_{x_3}$};
\node [style=node,fill=blue!30] (x5) at (2, 0) {$v_{x_2}$};
\node [style=node,fill=blue!30] (x6) at (1, 0) {$v_{x_1}$};

\node [style=node,fill=red!30] (c1) at (1, 3) {$w_{c_1}$};
\node [style=node,fill=red!30] (c2) at (2, 3) {$w_{c_2}$};
\node [style=node,fill=red!30] (c3) at (3, 3) {$w_{c_3}$};
\node [style=node,fill=red!30] (c4) at (4, 3) {$w_{c_4}$};
\node [style=node,fill=red!30] (c5) at (5, 3) {$w_{c_5}$};
\node [style=node,fill=red!30] (c6) at (6, 3) {$w_{c_6}$};

\node [draw=black,thick,dotted, shape=ellipse, minimum height=3cm, minimum width=9cm, label=above: $R_C$] (rc) at (3.5, 3) {};
\node [draw=black,thick,dotted, shape=ellipse, minimum height=3cm, minimum width=9cm, label=below: $B_V$] (bv) at (3.5, 0) {};
\node [draw=black, shape=circle, minimum size=3cm,fill=red!30, label=above: $R_D$] (rd) at (8, 4) {};
\node [draw=black, shape=circle, minimum size=3cm,fill=blue!30, label=above: $B_D$] (bd) at (-1, 4) {};
\node [draw=black, shape=circle, minimum size=2cm,fill=red!30, label=below: $R_E$] (re) at (8, 1) {};
\node [draw=black, shape=circle, minimum size=2cm,fill=blue!30, label=below: $B_E$] (be) at (-1, 1) {};
\node [style=node] (v1) at (-1, 3.6) {};
	\node [style=node] (v2) at (-1.4, 3.9) {};
	\node [style=node] (v3) at (-1.2, 4.3) {};
	\node [style=node] (v4) at (-0.8, 4.3) {};
	\node [style=node] (v5) at (-0.6, 3.9) {};
    \draw [style=edge] (v1) -- (v2);
    \draw [style=edge] (v1) -- (v3);
    \draw [style=edge] (v1) -- (v4);
    \draw [style=edge] (v1) -- (v5);
    \draw [style=edge] (v2) -- (v3);
    \draw [style=edge] (v2) -- (v4);
    \draw [style=edge] (v2) -- (v5);
    \draw [style=edge] (v3) -- (v4);
    \draw [style=edge] (v3) -- (v5);
    \draw [style=edge] (v4) -- (v5);
    \node [style=node] (u1) at (-1, 0.6) {};
	\node [style=node] (u2) at (-1.4, 0.9) {};
	\node [style=node] (u3) at (-1.2, 1.3) {};
	\node [style=node] (u4) at (-0.8, 1.3) {};
	\node [style=node] (u5) at (-0.6, 0.9) {};
    \draw [style=edge] (u1) -- (u2);
    \draw [style=edge] (u1) -- (u3);
    \draw [style=edge] (u1) -- (u4);
    \draw [style=edge] (u1) -- (u5);
    \draw [style=edge] (u2) -- (u3);
    \draw [style=edge] (u2) -- (u4);
    \draw [style=edge] (u2) -- (u5);
    \draw [style=edge] (u3) -- (u4);
    \draw [style=edge] (u3) -- (u5);
    \draw [style=edge] (u4) -- (u5);
\node [style=node] (v1) at (8, 3.6) {};
	\node [style=node] (v2) at (8.4, 3.9) {};
	\node [style=node] (v3) at (8.2, 4.3) {};
	\node [style=node] (v4) at (7.8, 4.3) {};
	\node [style=node] (v5) at (7.6, 3.9) {};
 \node [style=node] (u1) at (8, 0.6) {};
	\node [style=node] (u2) at (8.4, 0.9) {};
	\node [style=node] (u3) at (8.2, 1.3) {};
	\node [style=node] (u4) at (7.8, 1.3) {};
	\node [style=node] (u5) at (7.6, 0.9) {};

\draw [style=edge] (x1) -- (c4);
\draw [style=edge] (x1) -- (c5);
\draw [style=edge] (x1) -- (c6);
\draw [style=edge] (x2) -- (c3);
\draw [style=edge] (x2) -- (c5);
\draw [style=edge] (x2) -- (c6);
\draw [style=edge] (x3) -- (c2);
\draw [style=edge] (x3) -- (c4);
\draw [style=edge] (x3) -- (c6);
\draw [style=edge] (x4) -- (c1);
\draw [style=edge] (x4) -- (c4);
\draw [style=edge] (x4) -- (c5);
\draw [style=edge] (x5) -- (c1);
\draw [style=edge] (x5) -- (c2);
\draw [style=edge] (x5) -- (c3);
\draw [style=edge] (x6) -- (c1);
\draw [style=edge] (x6) -- (c2);
\draw [style=edge] (x6) -- (c3);
\draw [style=edge] (x1) edge [bend left] (x2);
\draw [style=edge] (x1) edge [bend left] (x3);
\draw [style=edge] (x1) edge [bend left] (x4);
\draw [style=edge] (x1) edge [bend left] (x5);
\draw [style=edge] (x1) edge [bend left] (x6);
\draw [style=edge] (x2) edge [bend left] (x3);
\draw [style=edge] (x2) edge [bend left] (x4);
\draw [style=edge] (x2) edge [bend left] (x5);
\draw [style=edge] (x2) edge [bend left] (x6);
\draw [style=edge] (x3) edge [bend left] (x4);
\draw [style=edge] (x3) edge [bend left] (x5);
\draw [style=edge] (x3) edge [bend left] (x6);
\draw [style=edge] (x4) edge [bend left] (x5);
\draw [style=edge] (x4) edge [bend left] (x6);
\draw [style=edge] (x5) edge [bend left] (x6);
\draw [style=edge, ultra thick] (rd) -- (rc);
\draw [style=edge, ultra thick] (bd) -- (rc);
\draw [style=edge, ultra thick] (be) edge [bend right] (rd);
\draw [style=edge, ultra thick] (be) -- (bd);
\draw [style=edge, ultra thick] (re) -- (rc);
\draw [style=edge, ultra thick] (re) -- (bv);
\draw [style=edge, ultra thick] (be) -- (rc);
\draw [style=edge, ultra thick] (be) -- (bv);
\end{tikzpicture}
    \caption{An example of a network $G$ which reduces the $3-XSAT^3_+$ problem $(x_1\lor x_2 \lor x_3) \land (x_1 \lor x_2 \lor x_4) \land (x_1 \lor x_3 \lor x_4) \land (x_2\lor x_5 \lor x_6) \land (x_3 \lor x_5 \lor x_6) \land (x_4 \lor x_5 \lor x_6)$ to $\frac14$-MIE and $\frac14$-MIAE}
    \label{fig:hardnessMIE+MIAE}
\end{figure}
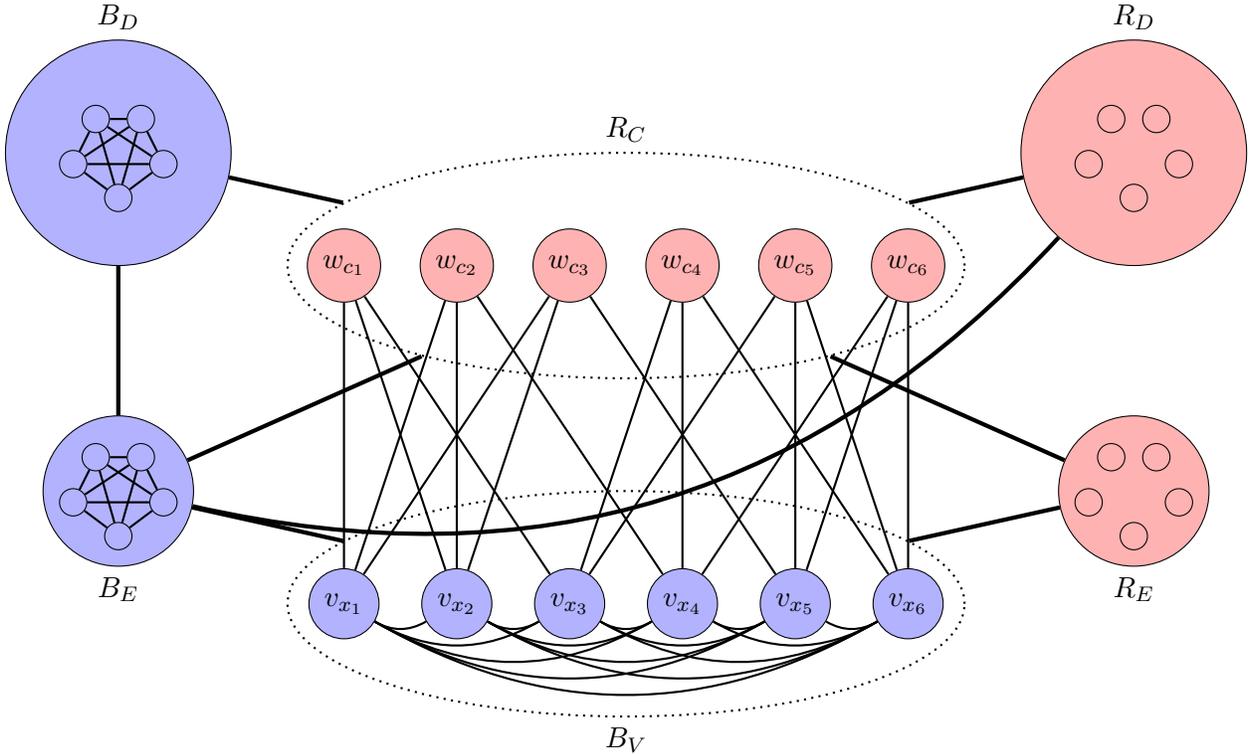

Now that we've constructed $G$, consider the $p$-I problem on $G$. Assume there is a solution to the instance of XSAT for $l$-CNF$_+^l$. There is therefore a set of $\frac{n}{l}$ variables set to TRUE in the solution, each belonging to $l$ clauses. 
For each of the true variables $x$, add an edge between $b_x$ and $r_C$ for each clause $C$ containing $x$. The result of this addition of $n$ edges is that all red nodes are no longer under illusion in $G$. However, there are now $\frac{n}{l}$ blue nodes corresponding to the TRUE variables which each have $l$ new red neighbours. They each require $a$ new blue neighbours to maintain the blue to red neighborhood ratio $\frac{a}{l}$. 
We therefore add $a$ edges to each of them, connecting them to some of the nodes of $B_D$, for a total of $\frac{na}{l}$ new edges. This fully eliminates illusion using a total of $\frac{n}{l}(a+l)$ edges.

Conversely, assume there is a solution to the $p$-I problem on $G$ which uses $\frac{n}{l}(a+l)$ edges.
Let $x_r,y_r$ denote the number of red neighbours removed in the solution from $x \in R_C$, $y \in B_V$ respectively. Let  $x_b,y_b$ denote the number of blue neighbours added in the solution from $x \in R_C$, $y \in B_V$ respectively. We know that each $x \in R_C$ begins under illusion, which the solution can eliminate by adding a blue neighbour or removing $ \lceil \frac{l}{a}\rceil$ red neighbours. Thus $$x_b + \left\lfloor  \frac{a}{l}\cdot x_r \right\rfloor =1$$ for each $x \in R_C$. As there are $n$ clauses, and the clause nodes form an independent set, we have

 \begin{equation}\label{eq:MIE-e1}
     n=\sum_{x \in R_C}\left( x_b + \left\lfloor \frac{a}{l}\cdot x_r\right\rfloor \right)
 \end{equation} 

 Next, if a clause node adds a blue neighbour, it must be a variable node, as each clause node is adjacent to all other blue nodes. The variable nodes have $b(b_x)=\frac{a}{l}r(b_x)$ for all $b_x \in B_V$, therefore adding new red neighbours increases illusion, which can be removed by removing a different red neighbour, or adding more blue neighbours. This gives us the following equation:

 \begin{equation}\label{eq:MIE-e2}
     \sum_{x \in R_C} x_b = \sum_{y \in B_V} \left(y_r + \left\lfloor \frac{l}{a}\cdot y_b \right\rfloor\right)
 \end{equation} 

Summing Equations \ref{eq:MIE-e1} and \ref{eq:MIE-e2} gives us

 \begin{equation}\label{eq:MIE-e3}
     n=\sum_{x \in R_C} \left\lfloor \frac{a}{l}\cdot x_r\right\rfloor + \sum_{y \in B_V} \left(y_r + \left\lfloor \frac{l}{a}\cdot y_b \right\rfloor\right)
 \end{equation} 
 
 By initial assumption, there exists a solution which uses $\frac{n}{l}(a+l)$ edges. This means 

\begin{equation}\label{eq:MIE-e4}
    \sum_{x \in R_C} (x_r + x_b) + \sum_{y \in B_V} (y_r + y_b) \leq \frac{n}{l}(a+l)
\end{equation}

If we subtract Equations \ref{eq:MIE-e1} and $\frac{a}{l}\cdot$\ref{eq:MIE-e3} from Inequality \ref{eq:MIE-e4}, this gives us

\begin{equation}\label{eq:MIE-e5}
    \sum_{x \in R_C}\left( x_r - \frac{l+a}{l}\cdot\left\lfloor \frac{a}{l}\cdot x_r\right\rfloor \right)) + \sum_{y \in B_V} \left(\frac{l-a}{l}\cdot y_r + y_b -\frac{a}{l}\cdot \left\lfloor \frac{l}{a}\cdot y_b \right\rfloor \right)\leq 0
\end{equation}

Recall that 
$$x_b + \left\lfloor \frac{a}{l}\cdot x_r \right\rfloor =1$$ 
for all $x \in R_C$, meaning $x_r \in \{0,\lceil \frac{l}{a}\rceil\}$. This means the terms $\sum_{x \in R_C}\left( x_r - \frac{l+a}{l}\cdot\lfloor \frac{a}{l}\cdot x_r\rfloor \right))$, $\sum_{y \in B_V}\frac{l-a}{l}\cdot y_r$ and $\sum_{y \in B_V} \left( y_b -\frac{a}{l}\cdot \lfloor \frac{l}{a}\cdot y_b \rfloor \right)$ are all non-negative. They must all be equal to 0 to satisfy Inequality \ref{eq:MIE-e4}. 
 This requires $x_r =0$ for all $x \in R_C$, $y_r =0$ for all $y \in B_V$, and finally $y_b =\frac{a}{l}\cdot \lfloor \frac{l}{a}\cdot y_b \rfloor$ for all $y \in B_V$. Note that this final equation implies each $y_b$ is a multiple of $a$ because $gcd(a,l)=1$. This simplifies Equations \ref{eq:MIE-e1} and \ref{eq:MIE-e3} to:

 \begin{equation}\label{eq:MIE-e6}
     n=\sum_{x \in R_C} x_b 
 \end{equation}

\begin{equation}\label{eq:MIE-e7}
     \frac{n}{l}=\sum_{y \in B_V} \frac{y_b}{a} 
 \end{equation} 
 
Together these form all $\frac{n(l+a)}{l}$ edges of the solution. This and the fact that each $y_b$ is a multiple of $a$ implies there are $\frac{n}{l}$ variable nodes which the solution connects to $l$ clauses each, so that each clause connects to exactly one of its variables. These correspond to the variables which satisfy the original X-SAT instance, as desired. 
\end{proof}

\begin{thm}
    The $p$-IR problem and $p$-I problem is NP-hard for $p\in \mathbb{Q}, \frac12 < p < 1, p \not=\frac23$.
\end{thm}

\begin{proof}
    First, $p = \frac{l}{b} > \frac12$ for $l,b \in \mathbb{N} \implies 2l > b$, therefore we can write $p = \frac{l}{b} = \frac{l}{a+l}$ for $l \geq 3, l> a, gcd(l,a) =1$.
    
    Consider an instance of XSAT for $l$-CNF$_+^l$ with $n$ variables and $n$ clauses. We build an auxiliary graph $G = (B\cup R, E)$  such that there is a solution to this XSAT for $l$-CNF$_+^l$ if and only if the $p$-I problem can be solved by removing $\frac{n}{l}\cdot (l+a)$ edges and adding no edges. Consequently, this also works as a reduction to $p$-IR. We design the graph $G$ to have nodes of seven types: blue variable nodes $B_V$, where $|B_V|=n$, and red clause nodes $R_C$, where $|R_C| = n$, as well as Dummy blue nodes $B_D^C$ for each clause $C$, where $|B_D^C| = ln-(n-1)$, Equalizing blue nodes $B_E$, where $|B_E| = l-1$, Dummy red nodes $R_D^C$ for each clause $C$, where $|R_D^C| = an-(l-1)$, Equalizing red nodes $R_E$, where $|R_E| = a$, and extra blue nodes $B'$, where $|B'| = n^4$. $G$ will have the following properties:

    \begin{itemize}
    \item {\tt Variable nodes}: The neighbourhood of $r_x$ for all $r_x \in R_V$ is $ R_E \cup \{b_c: b_c \in B_C, x \in C \text{ in the XSAT problem} \}$. Note that $r(r_x)=\frac{a}{l}b(r_x)$ for all $r_x \in R_V$.
    \item {\tt Clause nodes}: The neighbourhood of $b_c$ for all $b_c \in B_C$ is $ B_D^c \cup R_D^c \cup \{r_x: r_x \in R_V, x \in C \text{ in the XSAT problem} \} \cup B_C - b_c$. Note that   $r(b_c) = \frac{a}{l}\cdot b(b_c)+1$ for all $b_c \in B_C$.
    \item {\tt Dummy blue nodes}: The neighbourhood of $b_d$ for all $b_d \in B_D^c$ is $B' \cup \{b_c\}$. Note no node of $B_D$ is under illusion.
    \item {\tt Dummy red nodes}: The neighbourhood of $r_d$ for all $r_d \in R_D^c$ is $B_E \cup R_E \cup \{b_c\}$. Note  $r(r_d)=\frac{a}{l}b(r_d)$ for all $r_d \in R_D$.
    \item {\tt Equalizing blue nodes}: The neighbourhood of $b_e$ for all $b_e \in B_E$ is $R_D \cup B'$. Note no node of $B_E$ is under illusion.
    \item {\tt Equalizing red nodes}: The neighbourhood of $r_e$ for all $r_e \in R_E$ is $B' \cup R_V \cup \bigcup_{c\text{ is a clause}}R_D^c $. Note no node of $R_E$ is under illusion.
     \item {\tt Extra blue nodes}: The neighbourhood of $b'$ for all $b' \in B'$ is $B' \cup R_E \cup B_E$. Note no node of $B'$ is under illusion.
\end{itemize}

With $G$ constructed, the reader can see that its construction is analogous to that in \Cref{thm:MIEp<12}, and a nearly identical argument can be applied to complete the proof.
\end{proof}

\begin{thm}
    The $p$-IA problem is NP-hard for  $p\in \mathbb{Q}, \frac12 < p < 1, p \not=\frac23$.
\end{thm}

\begin{proof}
    First, $p = \frac{l}{b} > \frac12$ for $l,b \in \mathbb{N} \implies 2l > b$, therefore we can write $p = \frac{l}{b} = \frac{l}{a+l}$ for $l \geq 3, l> a, gcd(l,a) =1$.
    
    Consider an instance of XSAT for $l$-CNF$_+^l$ with $n$ variables and $n$ clauses, where $n \geq 2l$. We build an auxiliary graph $G = (B\cup R, E)$  such that there is a solution to this XSAT for $l$-CNF$_+^l$ if and only if the $p$-IA problem can be solved by adding $\frac{n}{l}\cdot (l+a)$ edges. We design the graph $G$ to have nodes of eight types: Dummy blue nodes $B_D$, where $|B_D| = ln^2-(l-1)n$, Dummy red nodes $R_D$, where $|R_D| = an^2-a(n-2)+\frac{n}{l}$, Equalizing red nodes $R_E$, where $|R_E| = a$, Equalizing blue nodes $B_E$, where $|B_E| = l-1$, blue variable nodes $B_V$, where $|B_V|=n$, blue clause nodes $B_C$, where $|B_C| = n$, extra red nodes $R'$, where $|R'| = a(n-2)$, and extra blue nodes $B'$, where $|B'| = (l-2)n-(l-1)$. $G$ will have the following properties:

    \begin{itemize}
    \item {\tt Variable nodes}: The neighbourhood of $b_x$ for all $b_x \in B_V$ is $B'\cup R' \cup \{b_c: b_c \in B_C, x \not\in C \text{ in the XSAT problem} \} \cup (B_V -b_x)$. Note that $b(b_x)= l(n-2) = \frac{l}{a}a(n-2) = \frac{l}{a}r(b_x)$ for all $b_x \in B_V$.
    \item {\tt Clause nodes}: The neighbourhood of $b_c$ for all $b_c \in R_C$ is $B'\cup R' \cup R_E \cup B_E \cup \{b_x: b_x \in B_V, x \not\in C \text{ in the XSAT problem} \} \cup B_C- b_c$. Note that   $b(b_c)= l(n-1)-1 = \frac{l}{a}a(n-1)-1 = \frac{l}{a}r(b_c)-1$ for all $b_c \in B_C$.
    \item {\tt Dummy blue nodes}: The neighbourhood of $b_d$ for all $b_d \in B_D$ is $R_E \cup B_E\cup R' \cup B' \cup B_D - b_d$. Note that no node of $B_D$ is under illusion.
    \item {\tt Dummy red nodes}: The neighbourhood of $r_d$ for all $r_d \in R_D$ is $R_E \cup B'$. Note no node of $R_D$ is under illusion.
    \item {\tt Equalizing blue nodes}: The neighbourhood of $b_e$ for all $b_e \in B_E$ is $B_D \cup B_C \cup B'$. Note no node of $B_E$ is under illusion.
    \item {\tt Equalizing red nodes}: The neighbourhood of $r_e$ for all $r_e \in R_E$ is $R' \cup B' \cup R_D \cup B_D \cup B_C \cup B_E$. Note each node $r_e$ requires $\frac{n}{l}$ new blue neighbours to correct its illusion.
    \item {\tt Extra blue nodes}: The neighbourhood of $b'$ for all $b' \in B'$ is $B_V \cup R_D \cup R_E \cup B_D \cup B_C \cup B_E \cup B' - b'$. Note no node of $B'$ is under illusion.
    \item {\tt Extra red nodes}: The neighbourhood of $r'$ for all $r' \in R'$ is $B_V \cup R_E \cup B_D \cup B_C \cup B_E$. Note no node of $R'$ is under illusion.
\end{itemize}
With $G$ constructed, the reader can see that any solution to $p$-IA on $G$ must add $\frac{an}{l}$ edges to correct the deficit of the red nodes. Given this, the proof that a solution using a total of $\frac{n(l+a)}{l}$ edges exists if and only if there is a solution to the X-SAT problem follows directly from the proof of  \Cref{thm:MIEp<12}.
\end{proof}

\begin{thm}
    The $p$-IR problem is NP-hard for  $p\in \mathbb{Q}, 0 < p< \frac12, p \not=\frac13$.
\end{thm}

\begin{proof}
    First, $p = \frac{a}{b} < \frac12$ for $a,b \in \mathbb{N} \implies 2a <  b$, therefore we can write $p = \frac{a}{b} = \frac{a}{a+l}$ for $l \geq 3, l> a, gcd(l,a) =1$.
    
    Consider an instance of XSAT for $l$-CNF$_+^l$ with $n$ variables and $n$ clauses. We build an auxiliary graph $G = (B\cup R, E)$  such that there is a solution to this XSAT for $l$-CNF$_+^l$ if and only if the $p$-IR problem can be solved by removing $\frac{n}{l}\cdot (l+a)$ edges. We design the graph $G$ to have nodes of eight types: Dummy blue nodes $B_D$, where $|B_D| = a$ Dummy red nodes $R_D$, where $|R_D| = 1$, Equalizing red nodes $R_E$, where $|R_E| = l$, Equalizing blue nodes $B_E$, where $|B_E| = a$, red variable nodes $R_V$, where $|R_V|=n$, red clause nodes $R_C$, where $|R_C| = n$, extra red nodes $R'$, where $|R'| = \frac{n}{l}$, and extra blue nodes $B'$, where $|B'| = n$. $G$ will have the following properties:

    \begin{itemize}
    \item {\tt Variable nodes}: The neighbourhood of $r_x$ for all $r_x \in R_V$ is $B_E \cup \{r_c: r_c \in R_C, x \in C \text{ in the XSAT problem} \}$. Note that $b(r_x)= a = \frac{a}{l}a = \frac{a}{l}r(r_x)$ for all $r_x \in R_V$.
    \item {\tt Clause nodes}: The neighbourhood of $r_c$ for all $r_c \in R_C$ is $B_D \cup \{r_x: r_x \in R_V, x \in C \text{ in the XSAT problem} \}$. Note that $r(r_c)= l+1 = \frac{l}{a}a+1 = \frac{l}{a}b(r_c)$ for all $r_c \in R_C$.
    \item {\tt Dummy blue nodes}: The neighbourhood of $b_d$ for all $b_d \in B_D$ is $R_C \cup B'$. Note that no node of $B_D$ is under illusion.
    \item {\tt Dummy red nodes}: The neighbourhood of $r_d$ for all $r_d \in R_D$ is $R_C \cup B'$. Note no node of $R_D$ is under illusion.
    \item {\tt Equalizing blue nodes}: The neighbourhood of $b_e$ for all $b_e \in B_E$ is $R' \cup R_V \cup B'$. Note each node $b_e$ requires $\frac{n}{l}$ fewer red neighbours to correct its illusion.
    \item {\tt Equalizing red nodes}: The neighbourhood of $r_e$ for all $r_e \in R_E$ is $R' \cup B'$.  Note no node of $R_E$ is under illusion.
    \item {\tt Extra blue nodes}: The neighbourhood of $b'$ for all $b' \in B'$ is $R_E \cup B_D \cup R_D \cup B_E \cup B' - b'$. Note no node of $B'$ is under illusion.
    \item {\tt Extra red nodes}: The neighbourhood of $r'$ for all $r' \in R'$ is $B_E \cup R_E$. Note $b(r')= a = \frac{a}{l}l = \frac{a}{l}r(r')$.
\end{itemize}
With $G$ constructed, the reader can see that any solution to $p$-IR on $G$ must add $\frac{an}{l}$ edges to correct the deficit of the red nodes. Given this, the proof that a solution using a total of $\frac{n(l+a)}{l}$ edges exists if and only if there is a solution to the X-SAT problem follows directly from the proof of  \Cref{thm:MIEp<12}.
\end{proof}

%% file: conclusion.tex
\section{Conclusion}\label{sec:conc}

We've settled the \textsc{$q$-Illusion Elimination} problems studied in \cite{Grandi2023} for the special $q=0$ case. While the previous paper showed hardness for $q=(0,1)$ and a trivial solution for $q=1$, we present here the surprising result that fully determining and eliminating illusion can be solved in polynomial time. The three problems, MIAE, MIRE, and MIE, have simple linear programs with integer solutions. 

We also define a natural extension to all three problems. For any $p \in [0,1]$, the $p$-IA, $p$-IR, and $p$-I problems ask if it is possible to alter a network to ensure the neighbourhood of every node has a $p$-fraction of blue nodes by adding edges, removing edges, or both respectively. These problems reduce to MIAE, MIRE and MIE when $p=\frac12$, and just as in those problems, we ask that the total number of blue nodes be such that the problems are solvable regardless of the initial network given. We show that these problems are hard for $p \in (0,\frac13)\cup (\frac13,\frac12)\cup (\frac12,\frac23)\cup (\frac23,1)$. 

In terms of future work, our $p$-IA, $p$-IR, and $p$-I problems have two $p$-values for which there is no polynomial time algorithm nor a hardness result: $p = \frac13$ and $p=\frac23$. Interestingly, if we add edge weights to the linear programs used to solve MIAE, MIRE and MIE, we can create linear programs to solve the $p$-IA, $p$-IR, and $p$-I problems. All values of $p$ except for $p = 0,\frac12, 1$ result in a program that is no longer totally dual integral. However, $p=\frac13,\frac23$ are the only other values for which the program is half integral. Therefore, it is entirely possible that for $p=\frac13,\frac23$, these problems are solvable in polynomial time, which would be an interesting result.